\documentclass[11pt]{article}

\usepackage{amssymb}
\usepackage{amsmath}
\usepackage{setspace}

\usepackage{algorithm}
\usepackage{algorithmic}

\usepackage{graphicx}
\usepackage{enumitem}
\usepackage{pgf,tikz}
\usepackage{subfigure}
\usetikzlibrary{arrows}
\usetikzlibrary{snakes}

\newtheorem{theorem}{Theorem}
\newtheorem{lemma}[theorem]{Lemma}

\newtheorem{problem}[theorem]{Problem}

 \newenvironment{proof}{\trivlist\item[]\emph{Proof}:}%
 {\unskip\nobreak\hskip 1em plus 1fil\nobreak$\Box$
 \parfillskip=0pt%
 \endtrivlist}

\newcommand{\qed}{}

\newcommand{\setofcycles}{\ensuremath{\mathcal{C}}}
\newcommand{\setofpaths}{\ensuremath{\mathcal{P}}}
\newcommand{\beadstring}{\ensuremath{B_{u,t}}}
\newcommand{\sbeadstring}{\ensuremath{B_{s,t}}}
\newcommand{\vbeadstring}{\ensuremath{B_{v,t}}}
\newcommand{\head}{\ensuremath{H_u}}

\newcommand{\chead}{\ensuremath{H_X}}

\newcommand{\chooseedge}{\ensuremath{\mathtt{choose}}}

\newcommand{\oracleleft}{\ensuremath{\mathtt{left\_update}}}
\newcommand{\oracleright}{\ensuremath{\mathtt{right\_update}}}
\newcommand{\undooracle}{\ensuremath{\mathtt{restore}}}

\newcommand{\bcc}{\textsc{bcc}}

\newcommand{\liststpaths}{\ensuremath{\mathtt{list\_paths}_{s,t}}}

\newcommand{\return}{\ensuremath{\mathtt{return}}}
\newcommand{\routput}{\ensuremath{\mathtt{output}}}

\begin{document}

%\vspace*{-2em}
\title{Optimal Listing of Cycles and st-Paths in Undirected Graphs}
\author{
% E. Birmel\'e (\small{U.~d'Evry, birmele@cnrs.fr}) 
%   \and  
  R. Ferreira (\small{U.~Pisa, ferreira@di.unipi.it}) 
  \and  
  R. Grossi (\small{U.~Pisa, grossi@di.unipi.it}) 
  \and  
  A. Marino (\small{U.~Firenze, a.marino@unifi.it}) 
  \and  
  N. Pisanti (\small{U.~Pisa, pisanti@di.unipi.it}) 
  \and  
  R. Rizzi (\small{U.~Verona, r.rizzi@univr.it}) 
  \and  
  G. Sacomoto (\small{INRIA, sacomoto@gmail.com}) 
}
%
% \institute{INRIA Grenoble Rh\^one-Alpes \&
% Universit\'e de Lyon 1, Villeurbanne, France\\
% \and
%  Universit\'e d'\'Evry, France
% \and
% Dipartimento di Informatica, Universit\`a di Pisa, Italy
% \and
% Dipartimento di Sistemi e Informatica, Universit\`a di Firenze, Italy
% \and
% Dipartimento di Informatica, Universit\`a di Verona, Italy
% }

\date{}

\maketitle
\setcounter{page}{0}
\thispagestyle{empty}
%\vspace*{-0.5em}
\begin{small}
  \centerline{\bf Abstract} 
  \medskip
  Listing all the simple cycles (hereafter just called cycles) in a
  graph is a classical problem whose efficient solutions date back to
  the early 70s. For a graph with $n$ vertices and $m$ edges,
  containing $\eta$ cycles, the best known solution in the literature
  is given by Johnson's algorithm [SIAM J. Computing, 1975] and takes
  $O((\eta+1)(m+n))$ time. This solution is surprisingly not optimal
  for undirected graphs: to the best of our knowledge, no
  theoretically faster solutions have been proposed in almost 40 years.

  We present the first optimal solution to list all the cycles in an
  undirected graph $G$, improving the time bound of Johnson's
  algorithm by a factor that can be $O(n^2)$.  Specifically, let
  $\setofcycles(G)$ denote the set of all these cycles, and observe
  that $|\setofcycles(G)| = \eta$. For a cycle $c \in
  \setofcycles(G)$, let $|c|$ denote the number of edges in~$c$. Our
  algorithm requires $O(m + \sum_{c \in \setofcycles(G)}{|c|})$ time
  and is asymptotically optimal: indeed, $\Omega(m)$ time is
  necessarily required to read $G$ as input, and $\Omega(\sum_{c \in
    \setofcycles(G)}{|c|})$ time is necessarily required to list the
  output.

  We then describe an infinite family of dense graphs, in which each
  graph with $n$ vertices and $m$ edges contains $\eta = \Omega(m)$
  cycles $c$ with $|c| = O(1)$ edges. For each graph in this family,
  our algorithm requires $\Theta(\eta+m) = \Theta(\eta)$ time, thus
  saving $O(m+n) = O(n^2)$ time when compared to Johnson's algorithm.
  In general for any graph, since $|c| \leq n$, the cost of our
  algorithm never exceeds $O(m + (\eta+1) n)$ time.

  We also present the first optimal solution to list all the simple
  paths from $s$ to $t$ (shortly, $st$-paths) in an undirected graph
  $G$. Let $\setofpaths_{st}(G)$ denote the set of $st$-paths in $G$
  and, for an $st$-path $\pi \in \setofpaths_{st}(G)$, let $|\pi|$ be
  the number of edges in $\pi$.  Our algorithm lists all the
  $st$-paths in~$G$ optimally in $O(m + \sum_{\pi \in
    \setofpaths_{st}(G)}{|\pi|})$ time, observing that
  $\Omega(\sum_{\pi \in \setofpaths_{st}(G)}{|\pi|})$ time is
  necessarily required to list the output.

%  \smallskip

  While the basic approach is simple (see binary partition in
  point~\ref{item:abstract:3}), we use a number of non-trivial ideas to
  obtain our optimal algorithm for an undirected (connected) graph $G$
  as conceptually listed below.
  \begin{enumerate}
  \item Prove the following reduction. If there exists an optimal
    algorithm to list the $st$-paths in $G$, there exists an optimal
    algorithm to list the cycles in $G$. This relates
    $\setofcycles(G)$ and $\setofpaths_{st}(G)$ for some choices $s,t$.

    % Consider a spanning tree of $G$ and the resulting non-tree
    % edges (shortly, back edges) $b_1, b_2, \ldots, b_r$. Let
    % $\setofcycles_i(G)$ denote the set of cycles in $G$ that contain
    % edge $b_i$ but do not contain $b_1, \ldots, b_{i-1}$. Note that
    % $\setofcycles_1(G)$,~\dots, $\setofcycles_r(G)$ are a partition of
    % $\setofcycles(G)$. Then, letting $b_i \equiv (s,t)$, the cycles
    % in $\setofcycles_i(G)$ can be listed as $st$-paths in $G \setminus
    % \{b_1, \ldots, b_i\}$. 

  \item Focus on listing the $st$-paths. Consider the decomposition of
    the graph into biconnected components ({\bcc}s), thus forming a
    tree $T$ where two {\bcc}s are adjacent in $T$ iff they share an
    articulation point. Exploit (and prove) the property that if $s$
    and $t$ belong to distinct {\bcc}s, then $(i)$ there is a unique
    \emph{sequence} $\sbeadstring$ of adjacent {\bcc}s in $T$ through
    which each $st$-path must necessarily pass, and $(ii)$ each
    $st$-path is the concatenation of paths connecting the
    articulation points of these {\bcc}s in $\sbeadstring$.

  \item \label{item:abstract:3} Recursively list the $st$-paths in
    $\sbeadstring$ using the classical binary partition (i.e.\mbox{}
    given an edge $e$ in $G$, list all the cycles containing
    $e$, and then all the cycles not containing~$e$): now it suffices to
    work on the \emph{first} \bcc\ in $\sbeadstring$, and efficiently
    maintain it when deleting an edge $e$, as required by the binary
    partition.

  \item Use a notion of \emph{certificate} to avoid recursive calls
    (in the binary partition) that do not list new $st$-paths.  This
    certificate is maintained dynamically as a data structure
    representing the first \bcc\ in $\sbeadstring$, which guarantees
    that there exists at least one \emph{new} solution in the current
    $\sbeadstring$.

  \item Consider the binary recursion tree corresponding to the binary
    partition.  Divide this tree into \emph{spines}: a spine
    corresponds to the recursive calls generated by the edges $e$
    belonging to the same adjacency list in $\sbeadstring$.  The
    amortized cost for each listed $st$-path $\pi$ is $O(|\pi|)$ if
    there is a guarantee that the amortized cost in each spine $S$ is
    $O(\mu)$, where $\mu$ is a lower bound on the number of $st$-paths
    that will be listed from the recursive calls belonging to $S$. The
    (unknown) parameter~$\mu$ different for each spine~$S$, and the
    corresponding cost $O(\mu)$, will drive the design of the proposed
    algorithms.
  \end{enumerate}
\end{small}

\newpage

\section{Introduction}
\label{sec:introduction}

%%% This is given in Section 2 - RG
% Let $G=(V,E)$ be an undirected graph with $n=|V|$ vertices and $m=|E|$
% edges. A \textit{walk} of length $k$ is a sequence of vertices $v_0,
% \ldots, v_{k}$ such that, for any $i$ with $0 \leq i < k$, $(v_i,
% v_{i+1})\in E$. A (simple) \textit{path} $\pi$ of length $k$ is a walk
% $v_0, \ldots, v_{k}$ such that, for any $i$ and $j$ with $0\leq i<j
% \leq k$, $v_i\neq v_j$, while a (simple) \textit{cycle} (or,
% equivalently, \textit{circuit}) $c$ of length $k+1$ is a path $v_0,
% \ldots, v_{k}$ such that $(v_k,v_0) \in E$.

Listing all the simple cycles (hereafter just called cycles) in a
graph is a classical problem whose efficient solutions date back to
the early 70s. For a graph with $n$ vertices and $m$ edges, containing
$\eta$ cycles, the best known solution in the literature is given by
Johnson's algorithm~\cite{Johnson1975} and takes
$O((\eta+1)(m+n))$ time. This solution is surprisingly not optimal for
undirected graphs: to the best of our knowledge, no theoretically faster solutions
have been proposed in almost 40 years.

%\smallskip
\textbf{Our results.}  
We present the first optimal solution to list all
the cycles in an undirected graph~$G$, improving the time bound of
Johnson's algorithm by a factor that can be $O(n^2)$.  Specifically,
let $\setofcycles(G)$ denote the set of all these cycles, and observe
that $|\setofcycles(G)| = \eta$. For a cycle $c \in \setofcycles(G)$,
let $|c|$ denote the number of edges in~$c$. Our algorithm requires
$O(m + \sum_{c \in \setofcycles(G)}{|c|})$ time and is asymptotically
optimal: indeed, $\Omega(m)$ time is necessarily required to read $G$
as input, and $\Omega(\sum_{c \in \setofcycles(G)}{|c|})$ time is
necessarily required to list the output. Since $|c| \leq n$, the cost
of our algorithm never exceeds $O(m + (\eta+1) n)$ time.

Along the same lines, we also present the first optimal solution to
list all the simple paths from $s$ to $t$ (shortly, $st$-paths) in an
undirected graph $G$. Let $\setofpaths_{st}(G)$ denote the set of
$st$-paths in $G$ and, for an $st$-path $\pi \in \setofpaths_{st}(G)$,
let $|\pi|$ be the number of edges in $\pi$.  Our algorithm lists
all the $st$-paths in~$G$ optimally in $O(m + \sum_{\pi \in
  \setofpaths_{st}(G)}{|\pi|})$ time, observing that $\Omega(\sum_{\pi
  \in \setofpaths_{st}(G)}{|\pi|})$ time is necessarily required to
list the output.  We prove the following reduction to relate
$\setofcycles(G)$ and $\setofpaths_{st}(G)$ for some suitable choices
of vertices $s,t$: If there exists an optimal algorithm to list the
$st$-paths in $G$, then there exists an optimal algorithm to list the
cycles in $G$.  Hence, we can focus on listing $st$-paths.

% We follow a binary partition approach: at each recursive step, we
% reduce to two subproblems, listing the paths that use an edge, and the
% ones that do not. Challenging here is to avoid subproblems that do not
% lead to any solutions. We maintain a certificate~\cite{Ferreira11}
% based on dynamic biconnected components and their DFS-tree: we show
% that the cost achieved to maintain the certificate is constant
% amortized with respect to the sum of the sizes of the solutions.

\textbf{History of the problem.}
The classical problem of listing all the cycles of a graph has been
extensively studied for its many applications in several fields,
ranging from the mechanical analysis of chemical
structures~\cite{Sussenguth65} to the design and analysis of reliable
communication networks, and the graph isomorphism
problem~\cite{Welch66}.
% In particular, at the turn of the seventies
% several algorithms for enumerating all cycles of an undirected graph
% have been proposed. 
There is a vast body of work, and the
majority of the algorithms listing all the cycles can be divided into
the following three classes (see~\cite{Bezem87,Mateti76} for excellent
surveys).

%\smallskip
\textit{Search space algorithms.} 
According to this approach, cycles are looked for in an appropriate
search space.  In the case of undirected graphs, the \emph{cycle vector
space} \cite{Diestel} turned out to be the most promising choice: from
a basis for this space, all vectors are computed and it is tested whether
they are a cycle. Since the algorithm introduced in~\cite{Welch66}, many
algorithms have been proposed: however, the complexity of these
algorithms turns out to be exponential in the dimension of the vector
space, and thus in $n$.  For the special case of planar graphs,
in~\cite{Syslo81} the author was able to design an algorithm listing
all the cycles in $O((\eta + 1)n)$ time.

%\smallskip
\textit{Backtrack algorithms.} 
According to this approach, all paths are generated by
backtrack and, for each path, it is tested whether it is a
cycle. One of the first algorithms based on this approach is the one
proposed in~\cite{Tiernan70}, which is however exponential in $\eta$. By
adding a simple pruning strategy, this algorithm has been successively
modified in~\cite{Tarjan73}: it lists all the cycles in $O(nm(\eta+1))$
time. Further improvements were proposed
in~\cite{Johnson1975},~\cite{Szwarcfiter76}, and~\cite{Read75},
leading to $O((\eta+1)(m+n))$-time algorithms that work for both directed
and undirected graphs. 
% Apart from the algorithm in~\cite{Tiernan70},
% all the algorithms based on this approach are \textit{polynomial-time
%   delay}, that is, the time elapsed between the outputting of two
% cycles is polynomial in the size of the graph (more precisely, $O(nm)$
% in the case of the algorithm of~\cite{Tarjan73} and $O(m)$ in the case
% of the other three algorithms).

%\smallskip
\textit{Algorithms using the powers of the adjacency matrix.} 
This approach uses the so-called \emph{variable adjacency matrix}, that is,
the formal sum of edges joining two vertices. A non-zero element of
the $p$-th power of this matrix is the sum of all walks of length $p$:
hence, to compute all cycles, we compute the $n$th power of
the variable adjacency matrix. This approach is not very efficient
because of the non-simple walks. All algorithms based
on this approach (e.g.\mbox{} \cite{Ponstein66}
and~\cite{Yau67}) basically differ only on the way they avoid to
consider walks that are neither paths nor cycles.

%\smallskip
% Even if so many results have been obtained 
Almost 40 years after Johnson's algorithm~\cite{Johnson1975}, the
problem of efficiently listing all cycles of a graph is still an
active area of research
(e.g.~\cite{ourSPIRE2012,Halford04,Horvath04,Liu06,Sankar07,Wild08,Schott11}).  New
application areas have emerged in the last decade, such as
bioinformatics: for example, two algorithms for this problem have been
proposed in~\cite{Klamt06} and~\cite{Klamt09} while studying
biological interaction graphs. Nevertheless, no significant
improvement has been obtained from the theory standpoint: in
particular, Johnson's algorithm is still the theoretically most
efficient.

\textbf{Hard graphs for Johnson's algorithm.}
We now describe an infinite family of dense undirected graphs, in
which each graph with $n$ vertices and $m$ edges contains $\eta =
\Omega(m)$ cycles $c$ with $|c| = O(1)$ edges. Suppose w.l.o.g.\mbox{}
that $n$ is a multiple of 3, and consider a tripartite complete graph
$K^3_{n/3} = (V_1 \cup V_2 \cup V_3, E)$, where $V_1, V_2, V_3$ are
pairwise disjoint sets of vertices of size $n/3$ each, and $E$ is the
set of edges thus formed: for each choice of $x_1 \in V_1, x_2 \in
V_2, x_3 \in V_3$, edges $(x_1,x_2)$, $(x_2,x_3)$, and $(x_1,x_3)$
belong to $E$; and no edges exist that connect any two vertices within $V_i$,
for $i=1,2,3$. Note that each choice of $x_1 \in V_1, x_2 \in V_2, x_3
\in V_3$ in $K^3_{n/3}$ gives rise to a distinct cycle $c \equiv
(x_1,x_2), (x_2,x_3), (x_3,x_1)$ of length $|c|=3$, and there are no
other (simple) cycles. Thus, there are $\eta = (n/3)^3$ cycles in the
graph $K^3_{n/3}$, where $m = 3 (n/3)^2$.  For each graph
$K^3_{n/3}$ in this family, our algorithm requires $\Theta(\eta+m) =
\Theta(\eta) = \Theta(n^3)$ time to list the cycles, thus saving $O(m+n)
= O(n^2)$ time when compared to Johnson's algorithm that takes
$O(n^5)$ in this case.
%[FIGURA DI $K^3_{n/3}$???]

Note that the analysis of the time complexity of Johnson's algorithm
is not pessimistic and cannot match the one of our algorithm for
listing cycles.  For example, consider the sparse ``diamond'' graph
$D_n = (V, E)$ in Fig.~\ref{fig:johnsoncounter} with $n=2k+3$
vertices in $V = \{a,b,c, v_1, \ldots, v_k, u_1, \ldots, u_k\}$. There
are $m = \Theta(n)$ edges in $E = \{ (a,c)$, $(a,v_i)$, $(v_i,b)$,
$(b,u_i)$, $(u_i,c)$, for $1 \leq i \leq k\}$, and three kinds of
(simple) cycles:
(1)~$(a, v_i), (v_i, b), (b, u_j), (u_j, c), (c, a)$ for $1 \leq i, j
\leq k$;
(2)~$(a, v_i), (v_i, b), (b, v_j), (v_j, a)$ for $1 \leq i < j \leq
k$;
(3)~$(b, u_i), (u_i, c), (c, u_j), (u_j, b)$ for $1 \leq i < j \leq
k$,
totalizing $\eta = \Theta(n^2)$ cycles.
Our algorithm takes $\Theta(n + k^2) = \Theta(\eta) = \Theta(n^2)$
time to list these cycles.  On the other hand, Johnson's algorithm
takes $\Theta(n^3)$ time, and the discovery of the $\Theta(n^2)$ cycles
in~(1) costs $\Theta(k) = \Theta(n)$ time each: the backtracking
procedure in Johnson's algorithm starting at $a$, and passing through
$v_i$, $b$ and $u_j$ for some $i,j$, arrives at $c$: at that point, it
explores all the vertices $u_l$ $(l \neq i)$ even if they do not lead
to cycles when coupled with $a$, $v_i$, $b$, $u_j$, and $c$.

%\textbf{Applications.}
%
%Considering the elements $c \in \setofcycles(G)$ and $\pi \in
%\setofpaths_{s,t}(G)$ as sequences of vertices, we can output them
%directly in compressed front coding \emph{FC}, where sequences with a
%common prefix are reported consecutively and the common prefix is
%reported just once~\cite{Witten:1994:MGC:561620}. As a result, the
%time bounds become $O(m+ \mathit{FC}(\setofcycles(G)))$ for cycles and
%$O(m+ \mathit{FC}(\setofpaths_{s,t}(G)))$ for $st$-paths, where
%$\mathit{FC}(\setofcycles(G)) \leq \sum_{c \in \setofcycles(G)}{|c|}$
%and $\mathit{FC}(\setofpaths_{s,t}(G)) \leq \sum_{\pi \in
%  \setofpaths_{s,t}(G)}{|\pi|}$, which is useful for compressible
%output.

\section{Overview and Main Ideas}
\label{sec:overview}

% graphs and subgraphs

\subsection*{Preliminaries}

Let $G=(V,E)$ be an undirected connected graph with $n=|V|$ vertices
and $m=|E|$ edges, without self-loops or parallel edges. For a vertex
$u \in V$, we denote by $N(u)$ the neighborhood of $u$ and by
$d(u)=|N(u)|$ its degree.  $G[V']$ denotes the subgraph \emph{induced}
by $V' \subseteq V$, and $G - u$ is the induced subgraph $G[ V
\setminus \{u\}]$ for $u \in V$. Likewise for edge $e \in E$, we adopt
the notation $G-e = (V,E \setminus \{e\})$.

Paths are simple in $G$ by definition: we refer to a path $\pi$ by its
natural sequence of vertices or edges.  A path $\pi$ from $s$ to $t$,
or $st$-\emph{path}, is denoted by $\pi = s \leadsto t$. Additionally,
$\setofpaths(G)$ is the set of all paths in $G$ and
$\setofpaths_{s,t}(G)$ is the set of all $st$-paths in $G$.  When
$s=t$ we have cycles, and $\setofcycles(G)$ denotes the set of all
cycles in $G$. We denote the number of edges in a path $\pi$ by
$|\pi|$ and in a cycle~$c$ by $|c|$. In this paper, we consider the following
problems.

% problems

\begin{problem}[Listing st-Paths]
	\label{prob:liststpaths}
	Given a graph $G=(V,E)$ and two distinct vertices $s,t \in V$,
        output all the paths $\pi \in \setofpaths_{s,t}(G)$.
\end{problem} 

\begin{problem}[Listing Cycles]
	\label{prob:listcycles}
	Given a graph $G=(V,E)$, output all the cycles $c \in \setofcycles(G)$.
\end{problem} 

% We say that an
% algorithm is \emph{optimal} if it takes $\smash{O(m + \sum_{\pi \in
% \setofpaths_{s,t}(G)}{|\pi|})}$ time (resp. $\smash{O(m + \sum_{c \in
% \setofcycles(G)}{|c|})}$) to solve Problem~\ref{prob:liststpaths}
% (resp. Problem~\ref{prob:listcycles}), since this is the time taken to
% read the input graph plus the time to list the output.
%, namely, constant time per edge in each of the paths
%in $\setofpaths_{s,t}(G)$ or the cycles in $\setofcycles(G)$.  

Our algorithms assume without loss of generality that the input graph
$G$ is connected, hence $m \ge n-1$, and use the decomposition of $G$
into biconnected components. Recall that an \emph{articulation point}
(or cut-vertex) is a vertex $u \in V$ such that the number of
connected components in $G$ increases when $u$ is removed. $G$ is
\emph{biconnected} if it has no articulation points. Otherwise, $G$
can always be decomposed into a tree of biconnected components, called
the \emph{block tree}, where each biconnected component is a maximal
biconnected subgraph of $G$ (see Fig.~\ref{fig:beadstring}), and
two biconnected components are adjacent if and only if they share an
articulation point.

% ADD THIS SUBSECTION?
%\subsection{From cycles to \boldmath{$st$}-paths}

\subsection{Reduction to \boldmath{$st$}-paths}
\label{sub:reduction-paths}

We now show that listing cycles reduces to listing $st$-paths while
preserving the optimal complexity.  

%%% Now this is unrelated to the proof
% Using the idea of performing a DFS traversal of $G$ to identify its
% back edges $b_1, \ldots, b_r$ and, for $i=1, \ldots, r$, listing all
% the cycles that contain $b_i$ but not $b_1,\ldots, b_{i-1}$, we have
% the following relation between the two problems.

\begin{lemma}
  \label{lemma:reduction}
  Given an algorithm that solves Problem~\ref{prob:liststpaths} in
  optimal $O(m + \sum_{\pi \in \setofpaths_{s,t}(G)}{|\pi|})$ time, there
  exists an algorithm that solves Problem~\ref{prob:listcycles} in
  optimal $O(m + \sum_{c \in \setofcycles(G)}{|c|})$ time.
\end{lemma}
\begin{proof}
  Compute the biconnected components of $G$ and keep them in a list
  $L$. Each (simple) cycle is contained in one of the biconnected
  components and therefore we can treat each biconnected component
  individually as follows. While $L$ is not empty, extract a biconnected
  component $B=(V_{B},E_{B})$ from $L$ and repeat the following three
  steps: $(i)$ compute a DFS traversal of $B$ and take any back edge
  $b=(s,t)$ in $B$; $(ii)$ list all $st$-paths in $B-b$, i.e.~the
  cycles in $B$ that include edge~$b$; $(iii)$ remove edge $b$ from
  $B$, compute the new biconnected components thus created by removing
  edge~$b$, and append them to $L$. When $L$ becomes empty, all the
  cycles in $G$ have been listed.

  Creating $L$ takes $O(m)$ time. For every $B \in L$, steps $(i)$ and
  $(iii)$ take $O(|E_B|)$ time.  Note that step $(ii)$ always outputs
  distinct cycles in $B$ (i.e.~$st$-paths in $B-b$) in
  $O(|E_{B}|+\sum_{\pi \in \setofpaths_{s,t}(B-b)}{|\pi|})$ time.
  However, $B-b$ is then decomposed into biconnected components whose
  edges are traversed again. We can pay for the latter cost: for any
  edge $e \neq b$ in a biconnected component $B$, there is always a
  cycle in $B$ that contains both $b$ and $e$ (i.e.\mbox{} it is an
  $st$-path in $B-b$), hence $\sum_{\pi \in
    \setofpaths_{s,t}(B-b)}{|\pi|}$ dominates the term $|E_{B}|$,
  i.e.~$\sum_{\pi \in \setofpaths_{s,t}(B-b)}{|\pi|}=
  \Omega(|E_{B}|)$.  Therefore steps $(i)$--$(iii)$ take $O(\sum_{\pi
    \in \setofpaths_{s,t}(B-b)}{|\pi|})$ time. When $L$ becomes empty,
  the whole task has taken $O(m + \sum_{c \in \setofcycles(G)}{|c|})$
  time.
\end{proof}

\subsection{Decomposition into biconnected components}
\label{sec:decomposition}

\begin{figure}[t!]
\centering
\begin{minipage}[b]{0.3\textwidth}
\centering
\definecolor{cqcqcq}{rgb}{0,0,0}
\begin{tikzpicture}[scale=0.8,line cap=round,line join=round,>=triangle 45,x=1.0cm,y=1.0cm]
%\draw [color=cqcqcq,dash pattern=on 2pt off 2pt, xstep=1.0cm,ystep=1.0cm] (-4.3,-3.64) grid (11.7,6.3);
%\clip(-4.3,-3.64) rectangle (11.7,6.3);
\draw (-2,2)-- (-1,3);
\draw [dash pattern=on 5pt off 5pt] (-2,2)-- (-1,2.52);
\draw [dash pattern=on 5pt off 5pt] (-2,2)-- (-1,2);
\draw (-1,3)-- (0,2);
\draw [dash pattern=on 5pt off 5pt] (-1,2.52)-- (0,2);
\draw [dash pattern=on 5pt off 5pt] (-1,2)-- (0,2);
\draw (-2,2)-- (-1,1);
\draw (-1,1)-- (0,2);
\draw (0,2)-- (1,3);
\draw [dash pattern=on 5pt off 5pt] (0,2)-- (1.04,2.48);
\draw [dash pattern=on 5pt off 5pt] (0,2)-- (1,2);
\draw (0,2)-- (1,1);
\draw (1,3)-- (2,2);
\draw [dash pattern=on 5pt off 5pt] (1.04,2.48)-- (2,2);
\draw [dash pattern=on 5pt off 5pt] (1,2)-- (2,2);
\draw (1,1)-- (2,2);
\draw [dash pattern=on 5pt off 5pt] (-2,2)-- (-1,1.48);
\draw [dash pattern=on 5pt off 5pt] (-1,1.48)-- (0,2);
\draw [dash pattern=on 5pt off 5pt] (0,2)-- (1,1.44);
\draw [dash pattern=on 5pt off 5pt] (1,1.44)-- (2,2);
\draw [shift={(0,2)}] plot[domain=0:3.14,variable=\t]({1*2*cos(\t r)+0*2*sin(\t r)},{0*2*cos(\t r)+1*2*sin(\t r)});
\begin{footnotesize}
\fill [color=black] (-2,2) circle (1.5pt);
\draw[color=black] (-2,2) node[left] {$a$};
\fill [color=black] (-1,3) circle (1.5pt);
\draw[color=black] (-1,3) node[above] {$v_1$};
\fill [color=black] (-1,2) circle (1.5pt);
\fill [color=black] (-1,1) circle (1.5pt);
\draw[color=black] (-1,1) node[below] {$v_k$};
\fill [color=black] (-1,1.48) circle (1.5pt);
\fill [color=black] (-1,2.52) circle (1.5pt);
\fill [color=black] (0,2) circle (1.5pt);
\draw[color=black] (0,2) node[above] {$b$};
\fill [color=black] (1,3) circle (1.5pt);
\draw[color=black] (1,3) node[above] {$u_1$};
\fill [color=black] (1,2) circle (1.5pt);
\fill [color=black] (1.04,2.48) circle (1.5pt);
\fill [color=black] (1,1) circle (1.5pt);
\draw[color=black] (1,1) node[below] {$u_k$};
\fill [color=black] (2,2) circle (1.5pt);
\draw[color=black] (2,2) node[right] {$c$};
\fill [color=black] (1,1.44) circle (1.5pt);
\end{footnotesize}
\end{tikzpicture}
\caption{Diamond graph.}
\label{fig:johnsoncounter}
\end{minipage}
\begin{minipage}[b]{0.6\textwidth}
\centering
\begin{tikzpicture}
[nodeDecorate/.style={shape=circle,inner sep=1pt,draw,thick,fill=black},%
  lineDecorate/.style={-,dashed},%
  elipseDecorate/.style={color=gray!30},
  scale=0.45]
%bicon
\fill [elipseDecorate] (5,10) circle (2);
\fill [elipseDecorate] (9,10) circle (2);
\fill [elipseDecorate] (2,10) circle (1);
\fill [elipseDecorate] (-1,10) circle (2);
\fill [elipseDecorate,rotate around={-55:(-1,8)}] (-1,7) circle (1);

\draw (5,10) circle (2);
\draw (9,10) circle (2);
\draw (2,10) circle (1);
\draw (-1,10) circle (2);
\draw (5,7) circle (1);
\draw (9,7) circle (1);
\draw [rotate around={55:(-1,8)}] (-1,7) circle (1);
\draw [rotate around={-55:(-1,8)}] (-1,7) circle (1);
\draw (13,10) circle (2);
\draw (13,7) circle (1);
\draw (-4,10) circle (1);
\begin{footnotesize}
\node (7) at (-2,7) [nodeDecorate,label=above:$s$] {};
\node (14) at (9,11) [nodeDecorate,label=above:$t$] {};
\end{footnotesize}
%% nodes or vertices
\foreach \nodename/\x/\y in {
  0/7/10, 1/5/11,
  2/3/10, 3/1/10, 4/-3/10, 5/-1/10, 6/-1/8, 7/-2/7, 8/-1/11,
  9/0/7,
  11/5/10, 12/4/9, 13/5/8, 14/9/11, 15/11/10, 16/9/9,
  17/9/7 , 18/9/8, 50/5/7, 51/13/11, 52/13/8, 53/13/7, 54/-4/10}
{
  \node (\nodename) at (\x,\y) [nodeDecorate] {};
}

%%% edges or lines
\path
\foreach \startnode/\endnode in {6/7, 6/9, 5/6, 5/3, 5/8, 4/5, 3/2,
2/11, 1/11, 11/12, 11/0, 11/13, 13/50, 0/14, 14/15, 15/16, 16/18,
18/17, 15/51, 15/52, 51/52, 52/53, 54/4}
{
  (\startnode) edge[lineDecorate] node {} (\endnode)
};

%%% backedges
\path
\foreach \startnode/\endnode/\bend in { 8/3/20, 6/3/20, 4/8/10,
12/13/20, 1/0/20, 2/1/20, 13/0/20, 0/18/10}
{
  (\startnode) edge[lineDecorate, bend left=\bend] node {} (\endnode)
};

%\path
%\foreach \startnode/\endnode/\bend in { 0/2/55, 4/9/10, 14/16/10,
%3/6/45, 1/12/10, 0/18/45}
%{
%  (\startnode) edge[lineDecorate, bend right=\bend] node {} (\endnode)
%};
\end{tikzpicture}
\caption{Block tree of $G$ with bead string $\sbeadstring$ in gray.}
\label{fig:beadstring}
\end{minipage}
\end{figure}

We now focus on listing $st$-paths
(Problem~\ref{prob:liststpaths}). We use the decomposition of $G$ into
a block tree of biconnected components.  Given vertices $s,t$, define
its \emph{bead string}, denoted by $\sbeadstring$, as the unique
sequence of one or more adjacent biconnected components (the
\emph{beads}) in the block tree, such that the first one contains $s$
and the last one contains $t$ (see Fig.~\ref{fig:beadstring}): these
biconnected components are connected through articulation points,
which must belong to all the paths to be listed.

\begin{lemma}
  \label{lemma:beadstring}
  All the $st$-paths in $\setofpaths_{s,t}(G)$ are contained in the
  induced subgraph $G[\sbeadstring]$ for the bead string
  $\sbeadstring$. Moreover, all the articulation points in
  $G[\sbeadstring]$ are traversed by each of these paths.
\end{lemma}
\begin{proof}
  Consider an edge $e = (u,v)$ in $G$ such that $u \in \sbeadstring$
  and $v \notin \sbeadstring$. Since the biconnected components of a
  graph form a tree and the bead string $\sbeadstring$ is a path in
  this tree, there are no paths $v \leadsto w$ in $G-e$ for any $w \in
  \sbeadstring$ because the biconnected components in $G$ are maximal
  and there would be a larger one (a contradiction).
  Moreover, let $B_1, B_2, \ldots, B_r$ be the biconnected components
  composing $\sbeadstring$, where $s \in B_1$ and $t \in B_r$. If
  there is only one biconnected component in the path (i.e.~$r=1$),
  there are no articulation points in $\sbeadstring$.  Otherwise, all
  of the $r-1$ articulation points in $\sbeadstring$ are traversed by
  each path $\pi \in \setofpaths_{s,t}(G)$: indeed, the articulation
  point between adjacent biconnected components $B_i$ and $B_{i+1}$ is
  their only vertex in common and there are no edges linking $B_i$ and
  $B_{i+1}$.
\end{proof}

We thus restrict the problem of listing the paths in
$\setofpaths_{s,t}(G)$ to the induced subgraph $G[\sbeadstring]$,
conceptually isolating it from the rest of $G$. For the sake of
description, we will use interchangeably $\sbeadstring$ and
$G[\sbeadstring]$ in the rest of the paper.

\subsection{Binary partition scheme}
\label{sec:basic-scheme}

We list the set of $st$-paths in $\sbeadstring$, denoted by
$\setofpaths_{s,t}(\sbeadstring)$, by applying the binary partition
method (where $\setofpaths_{s,t}(G) = \setofpaths_{s,t}(\sbeadstring)$
by Lemma~\ref{lemma:beadstring}): We choose an edge $e = (s,v)$
incident to~$s$ and then list all the $st$-paths that include $e$ and
then all the $st$-paths that do not include $e$. Since we delete some
vertices and some edges during the recursive calls, we proceed as follows.

\emph{Invariant}: At a generic recursive step on vertex $u$
(initially, $u:=s$), let $\pi_s = s \leadsto u$ be the path discovered
so far (initially, $\pi_s$ is empty $\{\}$). Let $\beadstring$ be the
current bead string (initially, $\beadstring :=
\sbeadstring$). More precisely, $\beadstring$ is defined as follows:
$(i)$~remove from $\sbeadstring$ all the nodes in $\pi_s$ but $u$, and
the edges incident to $u$ and discarded so far; $(ii)$~recompute the
block tree on the resulting graph; $(iii)$~$\beadstring$ is the unique
bead string that connects $u$ to $t$ in the recomputed block tree.

\smallskip

\emph{Base case}: When $u=t$, output the $st$-path $\pi_s$ computed so
far. 

\smallskip \emph{Recursive rule}: Let $\setofpaths(\pi_s, u,
\beadstring)$ denote the set of $st$-paths to be listed by the current
recursive call. Then, it is the union of the following two disjoint
sets, for an edge $e=(u,v)$ incident to~$u$:
\begin{itemize}
\item {[left branching]} the $st$-paths in $\setofpaths(\pi_s \cdot e,
  v, \vbeadstring)$ that use $e$, where $\vbeadstring$ is the unique
  bead string connecting $v$ to $t$ in the block tree resulting from
  the deletion of vertex $u$ from $\beadstring$;
\item {[right branching]} the $st$-paths in $\setofpaths(\pi_s, u,
  \beadstring')$ that do \emph{not} use~$e$, where $\beadstring'$ is
  the unique bead string connecting $u$ to $t$ in the block tree
  resulting from the deletion of edge $e$ from $\beadstring$.
\end{itemize}

\noindent
Hence, $\setofpaths_{s,t}(\sbeadstring)$ (and so
$\setofpaths_{s,t}(G)$) can be computed by invoking $\setofpaths(\{\}, s,
\sbeadstring)$. The correctness and completeness of the above approach
is discussed in Section~\ref{sec:intro-cert}.

At this point, it should be clear why we introduce the notion of bead
strings in the binary partition. The existence of the partial path
$\pi_s$ and the bead string $\beadstring$ guarantees that there surely
exists at least one $st$-path. But there are two sides of the coin
when using $\beadstring$.

\emph{(1)}~One advantage is that we can avoid useless recursive calls:
If vertex $u$ has only one incident edge $e$, we just perform the left
branching; otherwise, we can safely perform both the left and right
branching since the \emph{first} bead in $\beadstring$ is always a
biconnected component by definition (and so there exist both an
$st$-path that traverses $e$ and one that does not traverse $e$).

\emph{(2)}~The other side of the coin is that we have to maintain the
bead string $\beadstring$ as $\vbeadstring$ in the left branching and
as $\beadstring'$ in the right branching by
Lemma~\ref{lemma:beadstring}. Note that these bead strings are surely
non-empty since $\beadstring$ is non-empty by induction (we only
perform either left or left/right branching when there are solutions by
point~\emph{(1)}).

To efficiently address point~\emph{(2)}, we need to introduce the notion of
certificate as described next.

\subsection{Introducing the certificate}
\label{sec:intro-cert}

Given the bead string $\beadstring$, we call the \emph{head} of
$\beadstring$, denoted by $\head$, the first biconnected component in
$\beadstring$, where $u \in \head$. Consider a DFS tree of
$\beadstring$ rooted at $u$ that changes along with $\beadstring$, and
classify the edges in $\beadstring$ as tree edges or back edges (no
cross edges since the graph is undirected).

To maintain $\beadstring$ (and so $\head$) during the recursive calls,
we introduce a \emph{certificate} $C$ (see Fig.~\ref{fig:Certificate}
for an example): It is a suitable data structure that uses the above
classification of the edges in $\beadstring$, and supports the
following operations, required by the binary partition scheme.
\begin{itemize}
	\setlength{\itemsep}{0pt} 
      \item $\chooseedge(C,u)$: returns an edge $e = (u,v)$ with $v
        \in \head$ such that $\pi_s \cdot (u,v) \cdot u \leadsto t$ is
        an $st$-path such that $u \leadsto t$ is inside
        $\beadstring$. Note that $e$ always exists since $\head$ is
        biconnected. Also, the chosen $v$ is the last one in DFS order
        among the neighbors of $u$: in this way, the (only) tree edge
        $e$ is returned when there are no back edges leaving from
        $u$.\footnote{As it will be clear in
          Sections~\ref{sec:recursion-amortization}
          and~\ref{sec:certificate}, this order facilitates the
          analysis and the implementation of the certificate.}
%\item $\del(e)$ and $\del(u)$: given the edge $e \in E$ or vertex
%	$u \in V$, respectively updates the graph as $G = (V,E-\{e\})$
%	or $G = (V-\{u\},E)$.
%\item $\undel(e)$ and $\undel(u)$: resets the modifications done to $G$ by
%	$\del(e)$ and $\del(u)$.
      \item $\oracleleft(C,e)$: for the given $e=(u,v)$, it obtains
        $\vbeadstring$ from $\beadstring$ as discussed in
        Section~\ref{sec:basic-scheme}. This implies updating also
        $\head$, $C$, and the block tree, since the recursion
        continues on~$v$. It returns bookkeeping information $I$ for
        what is updated, so that it is possible to revert to
        $\beadstring$, $\head$, $C$, and the block tree, to their
        status before this operation.
\item $\oracleright(C,e)$: for the given $e=(u,v)$, it obtains
        $\beadstring'$ from $\beadstring$ as discussed in
        Section~\ref{sec:basic-scheme}, which implies updating also $\head$,
  $C$, and the block tree. It returns bookkeeping information $I$ as
  in the case of $\oracleleft(C,e)$.
\item $\undooracle(C,I)$: reverts the bead string to $\beadstring$, the
  head $\head$, the certificate $C$, and the block tree, to their
  status before operation $I := \oracleleft(C,e)$ or $I :=
  \oracleright(C,e)$ was issued (in the same recursive call).
\end{itemize}

\begin{figure*}
\centering
\begin{minipage}{0.45\textwidth}
\begin{figure}[H]
\centering
\begin{tikzpicture}
[nodeDecorate/.style={shape=circle,inner sep=1pt,draw,thick,fill=black},%
  lineDecorate/.style={-,dashed},%
  elipseDecorate/.style={color=gray!30},
  scale=0.20]
%bicon
\draw (10,22) circle (9);
\draw (5,11.1) circle (3);

\node (s) at (10,34) [nodeDecorate,color=lightgray,label=above left:$s$] {};
\node (u) at (10,31) [nodeDecorate,label=above left:$u$] {};

\node (tp) at (6.2,13.9) [nodeDecorate] {};
\node (t) at (5.5,11) [nodeDecorate,label=below:$t$] {};

\path {
	(s) edge[snake,-,color=lightgray] node {\quad\quad$\pi_s$} (u)
	(u) edge node {} (tp)
	(tp) edge node {} (t)
};

\node (a) at (9.3,28) [nodeDecorate,] {};
\node (b) at (7.1,18) [nodeDecorate,] {};
\node (c) at (13,25) [nodeDecorate,] {};
\node (d) at (14.5,22) [nodeDecorate] {};
\node (e) at (11,22) [nodeDecorate] {};
\node (f) at (16,19) [nodeDecorate,] {};
\node (g) at (13,19) [nodeDecorate,] {};
\node (h) at (10,16) [nodeDecorate,] {};

\path {
	(a) edge node {} (c)
	(c) edge node {} (d)
	(d) edge node {} (f)
	(c) edge node {} (e)
	(d) edge node {} (g)
	(b) edge node {} (h)
};

\path {
	(u) edge[dashed,bend left=-40] node {} (tp)
	(f) edge[dashed,bend left=-40] node {} (u)
	(g) edge[dashed,bend left=10] node {} (c)
	(e) edge[dashed,bend left=-10] node {} (u)
	(d) edge[dashed,bend left=-20] node {} (u)
	(a) edge[dashed,bend left=5] node {} (h)
};

\end{tikzpicture}
\caption{Example certificate $C$}
\label{fig:Certificate}
\end{figure}
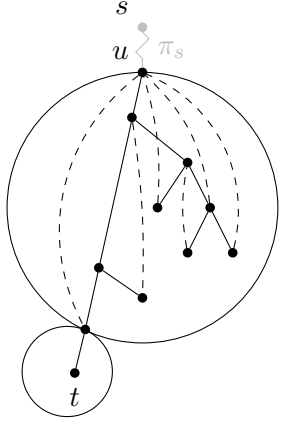
\end{minipage}
\hfill
\begin{minipage}{0.50\textwidth}

\begin{algorithm}[H]
	\caption{\label{alg:liststpaths} $\liststpaths(\pi_s,\,u,\,C)$}
	%%%%%\emph{Inv.:} $\pi_s = s \leadsto u$ and $\exists\, u \leadsto t$ by $C$
\begin{algorithmic}[1]
	\IF{$u=t$}
		\STATE $\routput(\pi_s)$ \label{code:base}
		\STATE $\return$ \label{code:returnbase}
	\ENDIF
	\STATE $e = (u,v) := \chooseedge( C, u )$ \label{code:choose}
	\IF{ $e \text{ is back edge}$ \label{code:if_back}}
		\STATE $I := \oracleright(C,e)$  \label{code:right_update}
		\STATE $\liststpaths(\pi_s,\, u,\,C)$ \label{code:right_branch}
		\STATE $\undooracle(C, I)$ \label{code:right_undo}
	\ENDIF
        \STATE $I := \oracleleft(C,e)$ \label{code:left_update}
        \STATE $\liststpaths( \pi_s \cdot (u,v),\, v,\, C)$ \label{code:left_branch}
        \STATE $\undooracle(C, I)$ \label{code:left_undo}
\end{algorithmic}
\end{algorithm}
\end{minipage}
\end{figure*}

Note that a notion of certificate in listing problems has been
introduced in~\cite{Ferreira11}, but it cannot be directly applied to
our case due to the different nature of the problems and our use of more
complex structures such as biconnected components. Using our certificate
and its operations, we can now formalize the binary partition and its
recursive calls $\setofpaths(\pi_s, u, \beadstring)$ described in
Section~\ref{sec:basic-scheme} as Algorithm~\ref{alg:liststpaths},
where $\beadstring$ is replaced by its certificate $C$.

The base case ($u=t$) corresponds to lines~1--4 of
Algorithm~\ref{alg:liststpaths}. During recursion, the left branching
corresponds to lines~5 and~11-13, while the right branching to
lines~5--10. Note that we perform only the left branching when there is
only one incident edge in $u$, which is a tree edge by definition of
$\chooseedge$. Also, lines~9 and~13 are needed to restore the
parameters to their values when returning from the recursive
calls. The proof of the following lemma is in the Appendix.

\begin{lemma}
  \label{lemma:correctness_algo_listpaths}
  %Given a correct implementation of the certificate $C$ and its
  %supported operations,
  Algorithm~\ref{alg:liststpaths} correctly lists all the $st$-paths in
  $\setofpaths_{s,t}(G)$.
\end{lemma}

A natural question is what is the complexity: we should account for
the cost of maintaining~$C$ and for the cost of the recursive calls of
Algorithm~\ref{alg:liststpaths}. Since we cannot always maintain the
certificate in $O(1)$ time, the ideal situation for attaining an
optimal cost is taking $O(\mu)$ time if at least $\mu$ $st$-paths are
listed in the current call (and its nested calls). An obstacle to this
ideal situation is that we cannot estimate~$\mu$ efficiently and
cannot design Algorithm~\ref{alg:liststpaths} so that it takes
$O(\mu)$ adaptively. We circumvent this by using a different cost
scheme in Section~\ref{sub:recursion-tree-cost} that is based on the
recursion tree induced by Algorithm~\ref{alg:liststpaths}.
Section~\ref{sec:certificate} and
Appendix~\ref{app:extend-analys-oper} are devoted to the efficient
implementation of the above certificate operations according to the cost
scheme that we discuss next.

\subsection{Recursion tree and cost amortization}
\label{sub:recursion-tree-cost}

We now define how to distribute the costs among the several recursive
calls of Algorithm~\ref{alg:liststpaths} so that optimality is
achieved. Consider a generic execution on the bead string
$\beadstring$. We can trace this execution by using a binary recursion
tree $R$.  The nodes of $R$ are labeled by the arguments on which
Algorithm~\ref{alg:liststpaths} is run: specifically, we denote a node
in $R$ by the triple $x = \langle \pi_s, u, C \rangle$ iff it
represents the call with arguments $\pi_s$, $u$, and~$C$.\footnote{For
  clarity, we use ``nodes'' when referring to $R$ and ``vertices''
  when referring to $\beadstring$.}  The left branching is represented
by the left child, and the right branching (if any) by the right child
of the current node. 
%We can observe the following properties.

\begin{lemma}
\label{lem:properties_recursion}
The binary recursion tree $R$ for $\beadstring$ has the following properties: 
\begin{enumerate}
	\setlength{\itemsep}{0pt} 
\item \label{item:R1} There is a one-to-one correspondence between the
  paths in $\setofpaths_{s,t}(\beadstring)$ and the leaves in the recursion
  tree rooted at node $\langle \pi_s, u, C \rangle$.
\item \label{item:R2} Consider any leaf and its corresponding $st$-path
  $\pi$: there are $|\pi|$ left branches in
  the corresponding root-to-leaf trace.
\item \label{item:R3} Consider the instruction $e:=\chooseedge(C,u)$
  in Algorithm~\ref{alg:liststpaths}: unary (i.e.\mbox{} single-child)
  nodes correspond to left branches ($e$ is a tree edge) while binary
  nodes correspond to left and right branches ($e$ is a back
  edge).
  %of $\beadstring$.
\item \label{item:R4} The number of binary nodes is
  $|\setofpaths_{s,t}(\beadstring)| - 1$.
\end{enumerate}
\end{lemma}

%Our costs need some notation and terminology on $R$. 

We define a \emph{spine} of $R$ to be a subset of $R$'s nodes linked
as follows: the first node is a node $x$ that is either the left child
of its parent or the root of $R$, and the other nodes are those
reachable from $x$ by right branching in $R$. Let $x = \langle \pi_s,
u, C \rangle$ be the first node in a spine $S$. The nodes in $S$
correspond to the edges that are incident to vertex $u$ in
$\beadstring$: hence their number equals the degree $d(u)$ of $u$ in
$\beadstring$, and the deepest (last) node in $S$ is always a tree
edge in $\beadstring$ while the others are back edges (in reverse DFS
order). Summing up, $R$ can be seen as composed by spines,
unary nodes, and leaves, where each spine has a unary node as deepest
node. This gives a global pictures of $R$ that we now exploit for the
analysis.

% In general, there is a one-to-one correspondence between the spines in
% $R$ and the bead strings and their heads obtained by
% Algorithm~\ref{alg:liststpaths}: as a result, for each spine, there is
% a unique corresponding $\beadstring$ and $\head$, a useful fact to
% prove our bounds. 

We define the \emph{compact head}, denoted by $\chead = (V_X, E_X)$,
as the (multi)graph obtained by compacting the maximal chains of
degree-2 vertices, except $u$, $t$, and the vertices that are the
leaves of its DFS tree rooted at $u$.

The rationale behind the above definition is that the costs defined in
terms of $\chead$ amortize well, as the size of $\chead$ and the
number of $st$-paths in the subtree of $R$ rooted at node $x = \langle
\pi_s, u, C \rangle$ are intimately related (see
Lemma~\ref{lemma:lower_bound_paths_beadstring} in
Section~\ref{sec:recursion-amortization}) while this is not
necessarily true for $\head$.

We now define the following abstract cost for spines, unary nodes, and
leaves of $R$, for a sufficiently large constant $c_0 > 0$, that
Algorithm~\ref{alg:liststpaths} must fulfill:
\begin{equation}
  \label{eq:abstrac_cost}
  T(r) =
  \left\{\begin{array}{ll}
      c_0 & \mbox{if $r$ is a unary node}\\ 
      c_0 |\pi| & \mbox{if $r$ is a leaf corresponding to path $\pi$}\\ 
      c_0 (|V_X|+|E_X|) \quad & \mbox{if $r$ is a spine with compact head $\chead$}
\end{array}\right. 
\end{equation}

\begin{lemma}
  \label{lemma:total_cost_recursion_tree}
  The sum of the costs in the nodes of the recursion tree $\sum_{r \in R} T(r) = O(\sum_{\pi \in \setofpaths_{s,t}(\beadstring)}{|\pi|})$.
\end{lemma}

Section~\ref{sec:recursion-amortization} contains the proof of
Lemma~\ref{lemma:total_cost_recursion_tree} and related properties.
Setting $u:=s$, we obtain that the cost in
Lemma~\ref{lemma:total_cost_recursion_tree} is optimal, by
Lemma~\ref{lemma:beadstring}.

\begin{theorem}
  \label{theorem:optimal_paths}
  Algorithm~\ref{alg:liststpaths} solves problem
  Problem~\ref{prob:liststpaths} in optimal $O(m + \sum_{\pi \in
  \setofpaths_{s,t}(G)}{|\pi|})$ time.
\end{theorem}

By Lemma~\ref{lemma:reduction}, we obtain
an optimal result for listing cycles.

\begin{theorem}
  \label{theorem:optimal_cycles}
  Problem~\ref{prob:listcycles} can be optimally solved in $O(m +
  \sum_{c \in \setofcycles(G)}{|c|})$ time. 
\end{theorem}

\section{Amortization strategy}
\label{sec:recursion-amortization}

We devote this section to prove
Lemma~\ref{lemma:total_cost_recursion_tree}. Let us split the sum
in Eq.~\eqref{eq:abstrac_cost} in three parts, and bound each part
individually, as
\begin{equation}
  \label{eq:sum_R}
  \sum_{r \in R} T(r) \leq \sum_{r:\,\mathrm{unary}} T(r) + \sum_{r:\,\mathrm{leaf}} T(r) + \sum_{r:\,\mathrm{spine}} T(r).
\end{equation}

We have that $\sum_{r:\,\mathrm{unary}} T(r) = O(\sum_{\pi \in
  \setofpaths_{s,t}(G)}{|\pi|})$, since there are
$|\setofpaths_{s,t}(G)|$ leaves, and the root-to-leaf trace leading to
the leaf for $\pi$ contains at most $|\pi|$ unary nodes by
Lemma~\ref{lem:properties_recursion}, where each unary node has cost
$O(1)$ by Eq.~\eqref{eq:abstrac_cost}.

Also, $\sum_{r:\,\mathrm{leaf}} T(r) = O(\sum_{\pi \in
  \setofpaths_{s,t}(G)}{|\pi|})$, since the leaf $r$ for $\pi$ has cost
$O(|\pi|)$ by Eq.~\eqref{eq:abstrac_cost}.

It remains to bound $\sum_{r\,\mathrm{spine}} T(r)$. By
Eq.~\eqref{eq:abstrac_cost}, we can rewrite this cost as
$\sum_{\chead} c_0 (|V_X| + |E_X|)$, where the sum ranges over the
compacted heads $\chead$ associated with the spines $r$. We use the
following lemma to provide a lower bound on the number of $st$-paths
descending from $r$.

\begin{lemma}
  \label{lemma:lower_bound_paths_beadstring}
  Given a spine $r$, and its bead string $\beadstring$ with head
  $\head$, there are at least $|E_X| - |V_X| + 1$ $st$-paths in $G$
  that have prefix $\pi_s = s \leadsto u$ and suffix $u \leadsto t$
  internal to $\beadstring$, where the compacted head is $\chead =
  (V_X, E_X)$.
\end{lemma}
\begin{proof}
   $\chead$ is biconnected. In any biconnected graph $B = (V_B, E_B)$
   there are at least $|E_B| - |V_B| + 1$ $xy$-paths for any $x,y \in
   V_B$. Find a ear decomposition \cite{Diestel} of $B$ and consider
   the process of forming $B$ by adding ears one at the time, starting
   from a single cycle including $x$ and $y$. Initially 
   $|V_B|=|E_B|$ and there are 2 $xy$-paths. Each new ear forms a path
   connecting two vertices that are part of a $xy$-path, increasing
   the number of paths by at least 1. If the ear has $k$ edges, its
   addition increases $V$ by $k-1$, $E$ by $k$, and the number of
   $xy$-paths by at least 1. The result follows by induction.
\end{proof}

The implication of Lemma~\ref{lemma:lower_bound_paths_beadstring} is
that there are at least $|E_X| - |V_X| + 1$ leaves descending from the
given spine $r$. Hence, we can charge to each of them a cost of
$\frac{c_0 (|V_X| + |E_X|)}{|E_X| - |V_X| +
  1}$. Lemma~\ref{lemma:density} allows us to prove that the latter
cost is $O(1)$ when $\head$ is different from a single edge or a
cycle. (If $\head$ is a single edge or a cycle, $\chead$ is a
  single or double edge, and the cost is trivially a constant.)
\begin{lemma}
  \label{lemma:density}
  For a compacted head $\chead = (V_X, E_X)$, its density is
  $\frac{|E_X|}{|V_X|} \geq \frac{11}{10}$.
\end{lemma}

Specifically, let $\alpha = \frac{11}{10}$ and write $\alpha = 1 +
2/\beta$ for a constant $\beta$: we have that $|E_X| + |V_X| = (|E_X|
- |V_X|) + 2 |V_X| \leq (|E_X| - |V_X|) + \beta (|E_X| - |V_X|) =
\frac{\alpha+1}{\alpha-1} (|E_X| - |V_X|)$. Thus, we can charge each
leaf with a cost of $\frac{c_0 (|V_X| + |E_X|)}{|E_X| - |V_X| + 1}
\leq c_0 \frac{\alpha+1}{\alpha-1} = O(1)$. This motivates the
definition of $\chead$, since Lemma~\ref{lemma:density} does not
necessarily hold for the head $\head$ (due to the unary nodes in its
DFS tree).

One last step to bound $\sum_{\chead} c_0 (|V_X| + |E_X|)$: as noted
before, a root-to-leaf trace for the string storing $\pi$ has $|\pi|$
left branches by Lemma~\ref{lem:properties_recursion}, and as many
spines, each spine charging $c_0 \frac{\alpha+1}{\alpha-1} = O(1)$ to
the leaf at hand. This means that each of the $|\setofpaths_{s,t}(G)|$
leaves is charged for a cost of $O(|\pi|)$, thus bounding the sum as
$\sum_{r\,\mathrm{spine}} T(r) = \sum_{\chead} c_0 (|V_X| + |E_X|) =
O(\sum_{\pi \in \setofpaths_{s,t}(G)}{|\pi|})$. This completes the
proof of Lemma~\ref{lemma:total_cost_recursion_tree}. As a corollary,
we obtain the following result.

\begin{lemma}
  \label{lemma:amortized_cost_per_path}
  The recursion tree $R$ with cost as in Eq.~\eqref{eq:abstrac_cost}
  induces an $O(|\pi|)$ amortized cost for each $st$-path $\pi$.
\end{lemma}

\section{Certificate implementation and maintenance}
\label{sec:certificate}

% We show how to represent and update the certificate $C$ so that the
% time taken by Algorithm~\ref{alg:liststpaths} in the recursion tree
% can be distributed among the nodes as in Eq.~\eqref{eq:abstrac_cost}:
% namely, (i)~$O(1)$ time in unary nodes; (ii)~$O(|\pi|)$ time in the
% leaf corresponding to path $\pi$; (iii)~$O(|V_X| + |E_X|)$ in each
% spine.

The certificate $C$ associated with a node $\langle \pi_s, u, C
\rangle$ in the recursion tree is a compacted and augmented DFS tree
of bead string $\beadstring$, rooted at vertex~$u$. The DFS tree
changes over time along with $\beadstring$, and is
maintained in such a way that $t$ is in the leftmost path of the tree.
% We augment the DFS tree with sufficient information to implicitly
% represent articulation points and biconnected components.  
% Recall that, by the properties of a DFS tree of an undirected graph,
% there are no cross edges (edges between different branches of the
% tree).
We compact the DFS tree by contracting the vertices that have
degree~2, except $u$, $t$, and the leaves (the latter surely have
incident back edges). Maintaining this compacted representation is not
a difficult data-structure problem. From now on we can assume
w.l.o.g.\mbox{} that $C$ is an augmented DFS tree rooted at $u$ where
internal nodes of the DFS tree have degree $\ge 3$, and each vertex
$v$ has associated:

\begin{enumerate}
	\setlength{\itemsep}{0pt} 
	\item A doubly-linked list $lb(v)$ of back edges linking 
		$v$ to its descendants $w$ sorted by DFS order.
	\item A doubly-linked list $ab(v)$ of back edges linking
          $v$ to its ancestors $w$ sorted by DFS order.
        \item \label{item:point3} An integer $\gamma(v)$, such that if $v$ is an
                ancestor of $w$ then $\gamma(v) < \gamma(w)$.
	\item \label{item:point4} The smallest $\gamma(w)$ over all neighbors $w$ of $v$ in $C$,
          excluding the parent, denoted by $\mathit{lowpoint}(v)$.
\end{enumerate}

Given two vertices $v,w \in C$ such that $v$ is the parent of $w$, we
can efficiently test if any of the children of $w$ is in the same
biconnected component of $v$, i.e.\mbox{} $\mathit{lowpoint}(w) \leq
\gamma(v)$. (Note that we adopt a variant of $\mathit{lowpoint}$ using
$\gamma(v)$ in place of $depth(v)$: it has the same effect whereas using
$\gamma(v)$ is preferable since it is easier to dynamically maintain.)

\begin{lemma}
  \label{lem:certificate_scratch}
 The certificate associated with the root of the recursion can be
 computed in $O(m)$ time.
\end{lemma}

\begin{proof}
	In order to set $t$ to be in the leftmost path, we perform a
	DFS traversal of graph $G$ starting from $s$ and stop when we
	reach vertex $t$. We then compute the DFS tree, traversing the
	path $s \leadsto t$ first. When visiting vertex $v$, we set
	$\gamma(v)$ to depth of $v$ in the DFS. Before going up
	on the traversal, we compute the lowpoints using the lowpoints
	of the children. Let $z$ be the parent of $v$. If $\mathit{lowpoint}(v)
	> \mathit{lowpoint}(z)$ and $w$ is not in the leftmost path in the DFS,
	we cut the subtree of $v$ as it does not belong to $B_{s,t}$.
	When finding a back edge $e=(v,w)$, if $w$ is a descendant of
	$v$ we append $e$ to both $lb(v)$ and $ab(w)$; else we append
	$e$ to both $ab(v)$ and $lb(w)$.  This maintains the DFS order
	in the back edge lists.  This procedure takes at most two DFS
	traversals in $O(m)$ time.  This DFS tree can be compacted in
	the same time bound.  \qed
\end{proof}

\begin{lemma}
	\label{lem:choose}
	Operation $\chooseedge(C,u)$ can be implemented in $O(1)$ time.
\end{lemma}
\begin{proof}
If the list $lb(v)$ is empty, return the tree edge $e=(u,v)$ linking $u$
to its only child $v$ (there are no other children).  Else,
return the last edge in $lb(v)$. \qed
\end{proof}

We analyze the cost of updating and restoring the
certificate $C$. We can reuse parts of~$C$,
namely, those corresponding to the vertices that are not in the
compacted head $H_X = (V_X,E_X)$ as defined in
Section~\ref{sub:recursion-tree-cost}.
%
% \begin{lemma}
% 	\label{lem:maintainsubtree}
%     Given a node $w \notin V_X$, the subtree of $w$ in the certificate
%     does not change after having performed a left or right branching.
% \end{lemma}
%
%
% Let us consider a unary node $\langle \pi_s, u, C \rangle$ of the
% recursion tree. Note that there are no back edges incident to $u$. 
We prove that, given a unary node $u$ and its tree edge $e=(u,v)$, the
subtree of $v$ in~$C$ can be easily made a certificate for the left
branch of the recursion.

\begin{lemma}
	\label{lem:unary_left}
	On a unary node, $\oracleleft(C,e)$ takes $O(1)$
	time.
\end{lemma}
\begin{proof}
	Take edge $e=(u,v)$. Remove edge $e$ and set $v$ as the root
	of the certificate. Since $e$ is the only edge incident in
	$v$, the subtree $v$ is still a DFS tree. Cut the list of children
	of $v$ keeping only the first child. (The other children are no
	longer in the bead string and become part of $I$.) There is
	no need to update $\gamma(v)$. \qed
\end{proof}

We now devote the rest of this section to show how to efficiently
maintain $C$ on a spine.  Consider removing a back edge $e$ from $u$:
the compacted head $H_X=(V_X,E_X)$ of the bead string can be divided
into smaller biconnected components.  Many of those can be excluded
from the certificate (i.e. they are no longer in the new bead string,
and so they are bookkept in $I$) and additionally we have to update
the lowpoints that change. We prove that this operation can be
performed in $O(|V_X|)$ total time on a spine of the recursion tree.

\begin{lemma}
	\label{lem:removebackedge}
	The total cost of all the operations $\oracleright(C,e)$ in a
        spine is $O(|V_X|)$ time.
\end{lemma}
\begin{proof}
	In the right branches along a spine, we remove all back edges
	in $lb(u)$. This is done by starting from the last edge in
	$lb(u)$, i.e.~proceeding in reverse DFS order.
	For back edge $b_i = (z_i,u)$, we traverse the vertices in the
	path from $z_i$ towards the root $u$, as these are the only lowpoints
	that can change.
	While moving upwards on the tree, on each vertex $w$, we
	update $\mathit{lowpoint}(w)$. This is done by taking the endpoint $y$
	of the first edge in $ab(w)$ (the back edge that goes the
	topmost in the tree) and choosing the minimum between
	$\gamma(y)$ and the lowpoint of each child of $w$. We stop
	when the updated $\mathit{lowpoint}(w) = \gamma(u)$ since it implies
	that the lowpoint of the vertex can not be further reduced.
	Note that we stop before $u$, except when removing the
	last back edge in $lb(u)$.

	To prune the branches of the DFS tree that are no longer in
	$B_{u,t}$, consider again each vertex $w$ in the path from
	$z_i$ towards the root $u$ and its parent $y$. It is possible
	to test if $w$ is an articulation point by checking if the
	updated $\mathit{lowpoint}(w) > \gamma(y)$. If that is the case and
	$w$ is not in the leftmost path of the DFS, it implies that $w
	\notin B_{u,t}$, and therefore we cut the subtree of $w$ and
	keep it in $I$ to restore later. We use the same halting
	criterion as in the previous paragraph.
	
	The cost of removing all back edges in the spine is
	$O(|V_X|)$: there are $O(|V_X|)$ tree edges and, in the paths
	from $z_i$ to $u$, we do not traverse the same tree edge twice
	since the process described stops at the first common ancestor
	of endpoints of back edges $b_i$. Additionally, we take $O(1)$
	time to cut a subtree of an articulation point in the DFS tree. \qed
\end{proof}

To compute $\oracleleft(C,e)$ in the binary nodes of a spine, we use
the fact that in every left branching from that spine, the graph is
the same (in a spine we only remove edges incident to $u$ and on a
left branch from the spine we remove the node $u$) and therefore its
block tree is also the same. However, the certificates
on these nodes are not the same, as they are rooted at different
vertices. By using the DFS order of the edges, we are able to traverse
each edge in $H_X$ only a constant number of times in the spine.

\begin{lemma}
	\label{lem:promotebackedge}
	The total cost of all operations $\oracleleft(C,e)$ in a
        spine is amortized $O(|E_X|)$.
\end{lemma}
\begin{proof}
	Let $t'$ be the last vertex in the path $u \leadsto t$ s.t.
	$t' \in V_X$. Since $t'$ is an articulation point, the subtree
        of the DFS tree rooted in $t'$ is maintained in the
	case of removal of vertex $u$. Therefore the only modifications
	of the DFS tree occur in the compacted head $H_X$ of $B_{u,t}$.
	Let us compute the certificate $C_i$: this is the certificate
	of the left branch of the $i$th node of the spine where we
	augment the path with the back edge $b_i = (z_i,u)$ of
	$lb(u)$ in the order defined by $\chooseedge(C,u)$.

	For the case of $C_1$, we remove $u$ and rebuild the
	certificate starting form $z_1$ (the last edge in $lb(u)$)
	using the algorithm from Lemma~\ref{lem:certificate_scratch}
	restricted to $H_X$ and using $t'$ as target and $\gamma(t')$
	as a baseline to $\gamma$ (instead of the depth). This takes
	$O(|E_X|)$.  

	For the general case of $C_i$ with $i>1$ we also rebuild
	(part) of the certificate starting from $z_i$ using the
	procedure from Lemma~\ref{lem:certificate_scratch} but we use
	information gathered in $C_{i-1}$ to avoid exploring useless
	branches of the DFS tree. The key point is that, when we reach
	the first bead in common to both $B_{z_i,t}$ and
	$B_{z_{i-1},t}$, we only explore edges internal to this bead.
	If an edge $e$ leaving the bead leads to $t$, we can reuse
	a subtree of $C_{i-1}$. If $e$ does not lead to $t$, then it
	has already been explored (and cut) in $C_{i-1}$ and there is
	no need to explore it again since it will be discarded.
	Given the order we take $b_i$, each bead is
	not added more than once, and the total cost over the
	spine is $O(|E_X|)$.

	Nevertheless, the internal edges $E_X'$ of the first bead in
	common between $B_{z_i,t}$ and $B_{z_{i-1},t}$ can be explored
	several times during this procedure.\footnote{Consider the
	case where $z_i, \ldots, z_j$ are all in the same bead after
	the removal of $u$. The bead strings are the same, but the
	roots $z_i, \ldots, z_j$ are different, so we have to compute
	the corresponding DFS of the first component $|j-i|$ times.}
	We can charge the cost $O(|E'_X|)$ of exploring those edges to
	another node in the recursion tree, since this common bead is
	the head of at least one certificate in the recursion subtree
	of the left child of the $i$th node of the spine.
	Specifically, we charge the first node in the \emph{leftmost}
	path of the $i$th node of the spine that has exactly the
	edges $E'_X$ as head of its bead string: (i) if $|E'_X| \le 1$
	it corresponds to a unary node or a leaf in the recursion tree
	and therefore we can charge it with $O(1)$ cost; (ii)
	otherwise it corresponds to a first node of a spine and
	therefore we can also charge it with $O(|E'_X|)$. We use this
	charging scheme when $i \neq 1$ and the cost is always charged
	in the leftmost recursion path of $i$th node of the spine.
	Consequently, we never charge a node in the recursion tree more
	than once.  \qed
\end{proof}

\begin{lemma}
	\label{lem:restore} On each node of the recursion tree,
	$\undooracle(C,I)$ takes time proportional to the size of the
	modifications kept in $I$.
\end{lemma}

From Lemmas~\ref{lem:choose} and
\ref{lem:removebackedge}--\ref{lem:restore}, it follows that on a
spine of the recursion tree we have the costs:
$\chooseedge(u)$ on each node which is bounded by $O(|V_X|)$ as there
are at most $|V_X|$ back edges in $u$; $\oracleright(C,e)$,
$\undooracle(C,I)$ take $O(|V_X|)$ time; $\oracleleft(C,e)$ and
$\undooracle(C,I)$ are charged $O(|V_X|+|E_X|)$ time.  We thus have the following
result,  completing the proof of Theorem~\ref{theorem:optimal_paths}.

% In Lemmas~\ref{lem:unarycost}, \ref{lem:spinecost}, \ref{lem:leafcost}
% we prove that the time requirements in Eq.~\eqref{eq:abstrac_cost} are
% met.  Together with Lemma~\ref{lemma:amortized_cost_per_path}, this
% finalizes the proof of Theorem \ref{theorem:optimal_paths} and
% \ref{theorem:optimal_cycles}. 

\begin{lemma}
	\label{lem:algo_cost}
	Algorithm~$\ref{alg:liststpaths}$ can be implemented with a
        cost fulfilling Eq.~\eqref{eq:abstrac_cost}, thus it takes
        total $O(m+\sum_{r \in R} T(r)) = O(m+\sum_{\pi \in
          \setofpaths_{s,t}(\beadstring)}{|\pi|})$ time.
\end{lemma}

\clearpage
%\nocite{*}
\begin{singlespace}
\begin{footnotesize}
\bibliographystyle{plain}
\bibliography{cycles}
\end{footnotesize}
\end{singlespace}

\clearpage
\appendix

\section{Omitted Proofs}
\label{app:omitted-proofs}

\subsection{Lemma~\ref{lemma:correctness_algo_listpaths}}
\begin{proof}
  For a given vertex $u$ the function $\chooseedge(C, u)$ returns an
  edge $e$ incident to $u$. We maintain the invariant that $\pi_s$ is
  a path $s \leadsto u$, since at the point of the recursive call in
  line~\ref{code:left_branch}: (i) is connected as we append edge
  $(u,v)$ to $\pi_s$ and; (ii) it is simple as vertex $u$ is removed
  from the graph $G$ in the call to $\oracleleft(C,e)$ in
  line~\ref{code:left_update}. In the case of recursive call in
  line~\ref{code:right_branch} the invariant is trivially maintained
  as $\pi_s$ does not change.
  The algorithm only outputs $st$-paths since $\pi_s$ is
  a $s \leadsto u$ path and $u=t$ when the algorithm outputs, in
  line~\ref{code:base}. 

  All the paths with prefix $\pi_s$ that do not use $e$ are listed by
  the recursive call in line~\ref{code:right_branch}. This is done by
  removing $e$ from the graph in line~\ref{code:right_update} and thus
  no path can include $e$. The paths that use $e$ are listed in
  line~\ref{code:left_branch} since in the recursive call $e$ is added
  to $\pi_s$. Given that the tree edge incident to $u$ is the last one
  to be returned by $\chooseedge(C,u)$, there is no path that does not
  use this edge, therefore it is not necessary to call
  line~\ref{code:right_branch} for this edge.
\end{proof}

\subsection{Lemma~\ref{lem:properties_recursion}}

\begin{proof}
We proceed in order as follows.
\begin{enumerate}
\item We only output a solution in a leaf and we only do recursive
  calls that lead us to a solution. Moreover every node partitions the
  set of solutions in the ones that use an edge and the ones that do
  not use it. This guarantees that the leaves in the left subtree of
  the node corresponding to the recursive call and the leaves in the
  right subtree do not intersect. This implies that different leaves
  correspond to different paths from $s$ to $t$, and that for each
  path there is a corresponding leaf.
\item Each left branch corresponds to the inclusion of an edge in the
  path $\pi$.
\item Since we are in a biconnected component, there is always a left
  branch. There can be no unary node as a right branch: indeed for any
  edge of $\beadstring$ there exists always a path from $s$ to $t$
  passing through that edge. Since the tree edge is always the last
  one to be chosen, unary nodes cannot correspond to back edges and
  binary nodes are always back edges.
\item From point \ref{item:R1} and from the fact that the recursion tree is
  a binary tree. (In any binary tree, the number of binary nodes is
  equal to the number of leaves minus 1.)
\end{enumerate}
\end{proof}

%\subsection{Lemma~\ref{lemma:compacted}}
%
%\begin{proof}
%  Let $\beadstring = (V_B, E_B)$ and $\cbeadstring = (V_X, E_X)$. There
%  is a one-to-one correspondence between the paths $\pi = x \leadsto
%  y$, for $x,y \in V_B \cap V_X$. In other words, there is a
%  one-to-one correspondence between paths with end points in the
%  non-compressed vertices of $\beadstring$. Note that $\cbeadstring$
%  is built from $\beadstring$ by replacing all maximal non-branching
%  paths by a single edge.
%  
%  Consider a path $\pi = x \leadsto y$ of $\beadstring$, with $x,y \in
%  V_B \cap V_X$. Assume for simplicity that there is just one vertex
%  $z$ with $d(z) \geq 3$. The subpaths $x \leadsto z$ and $z \leadsto
%  y$, are clearly non-branching. The maximality comes from the choice
%  of $x$ and $y$. For instance if $x \leadsto z$ were not maximal we
%  could extend it to the left and $x \notin V_X$. So the path $\pi$ of
%  $\beadstring$ corresponds to the path $\pi' = (x, z) (z, y)$ of
%  $\cbeadstring$.
%
%  In the other direction, consider a path $\pi'$ of
%  $\cbeadstring$. All vertices in $\pi'$ also belong to $\beadstring$,
%  since the compression does not create any new vertex. By replacing
%  every edge $e$ of $\pi'$ such that $e \in E_X \setminus E_B$, by the
%  corresponding maximal paths, we obtain the path $\pi$ of
%  $\beadstring$.
%\end{proof}

\subsection{Lemma~\ref{lemma:density}}

\begin{proof}
	Consider the following partition $V_X = \{r\} \cup V_2 \cup
	V_3$ where: $r$ is the root; $V_2$ is the set of vertices with
	degree 2 and; $V_3$, the vertices with degree $\geq 3$.  Since
	$H_X$ is compacted DFS tree of a biconnected graph, we have
	that $V_2$ is a \emph{subset} of the leaves and $V_3$ contains
	the set of internal nodes (except $r$). There are no vertices
	with degree 1 and $d(r) \geq 2$. Let $x = \sum_{v \in V_3}
	d(v)$ and $y = \sum_{v \in V_2} d(v)$.  We can write the
	density as a function of $x$ and $y$:
	$$\frac{|E_X|}{|V_X|} = \frac{x + y + d(r)}{2 (|V_3| + |V_2| + 1)} $$

	Note that $|V_3| \le \frac{x}{3}$ as the vertices in $V_3$
	have at least degree 3, $|V_2| = \frac{y}{2}$ as vertices in
	$V_2$ have degree exactly 2. Since $d(r) \ge 2$, we derive the
	following bound:
	$$\frac{|E_X|}{|V_X|} \ge \frac{x + y + 2}{\frac{2}{3}x + y + 2} $$

	Consider any graph with $|V_X|>3$ and its DFS tree rooted at
        $r$. Note that: (i) there are no tree edges between any two
        leaves, (ii) every node in $V_2$ is a leaf and (iii) no leaf
        is a child of $r$.  Therefore, every tree edge incident in a
        vertex of $V_2$ is also incident in a vertex of $V_3$. Since
        exactly half the incident edges to $V_2$ are tree edges (the
        other half are back edges) we get that $y \le 2x$.

	With $|V_X| \ge 3$ there exists at least one internal node in
	the DFS tree and therefore $x \ge 3$.
 \begin{alignat*}{2}
	 \text{minimize }\quad   & \frac{x + y + 2}{\frac{2}{3}x + y + 2}\ \\
	 \text{subject to }\quad  & 0 \le y \le 2x, \\
	 		    & x \ge 3.
  \end{alignat*}

 Since for any $x$ the function is minimized by the maximum $y$ s.t.
 $y \le 2x$ and for any $y$ by the minimum $x$, we get:
 $$\frac{|E_X|}{|V_X|} \ge \frac{9x + 6}{8x + 6} \ge
 \frac{11}{10}$$
 \qed
\end{proof}

\subsection{Lemma~\ref{lem:restore}}
\begin{proof}
	We use standard data structures (i.e. linked lists) for the
        representation of certificate $C$.  Persistent versions of
        these data structures exist that maintain a stack of
        modifications applied to them and that can restore its
        contents to their previous states.  Given the modifications in
        $I$, these data structures take $O(|I|)$ time to restore the
        previous version of $C$.

        Let us consider the case of performing $\oracleleft(C,e)$. We
        cut at most $O(|V_X|)$ edges from $C$. Note that, although we
        conceptually remove whole branches of the DFS tree, we only
        remove edges that attach those branches to the DFS tree. The
        other vertices and edges are left in the certificate but, as
        they no longer remain attached to $B_{u,t}$, they will never
        be reached or explored. In the case of $\oracleright(C,e)$, we
        have a similar situation, with at most $O(|E_X|$) edges being
        modified along the spine of the recursion tree. \qed
\end{proof}
	
\clearpage

\section{Extended analysis of operations in a spine of the recursion tree}
\label{app:extend-analys-oper}

\begin{figure}[t!]
\centering
%%%%%%%% ROOT
\subfigure[Bead string $B_{u,t}$ at the root of the spine]{
\begin{tikzpicture}
[nodeDecorate/.style={shape=circle,inner sep=1pt,draw,thick,fill=black},%
  lineDecorate/.style={-,dashed},%
  elipseDecorate/.style={color=gray!30},
  scale=0.35]
%bicon
\draw (10,22) circle (9);
\draw (5,11.1) circle (3);

\node (s) at (10,34) [nodeDecorate,color=lightgray,label=above left:$s$] {};
\node (u) at (10,31) [nodeDecorate,label=above left:{ $u$}] {};

\node (tp) at (6.2,13.9) [nodeDecorate,label=above right:{ $z_4$}] {};
\node (t) at (5.5,11) [nodeDecorate,label=below:$t$] {};

\path {
	(s) edge[snake,-,color=lightgray] node {\quad\quad$\pi_s$} (u)
	(u) edge node {} (tp)
	(tp) edge node {} (t)
};

\node (a) at (9.3,28) [nodeDecorate,label=left:$v$] {};
\node (b) at (7.1,18) [nodeDecorate,] {};
\node (c) at (13,25) [nodeDecorate,] {};
\node (d) at (14.5,22) [nodeDecorate,label=left:$z_2$] {};
\node (e) at (11,22) [nodeDecorate,label=below:$z_3$] {};
\node (f) at (16,19) [nodeDecorate,label=left:$z_1$] {};
\node (g) at (13,19) [nodeDecorate,] {};
\node (h) at (10,16) [nodeDecorate,] {};

\path {
	(a) edge node {} (c)
	(c) edge node {} (d)
	(d) edge node {} (f)
	(c) edge node {} (e)
	(d) edge node {} (g)
	(b) edge node {} (h)
};

\path {
	(u) edge[dashed,bend left=-40] node {} (tp)
	(f) edge[dashed,bend left=-40] node {} (u)
	(g) edge[dashed,bend left=10] node {} (c)
	(e) edge[dashed,bend left=-10] node {} (u)
	(d) edge[dashed,bend left=-20] node {} (u)
	(a) edge[dashed,bend left=5] node {} (h)
};

\end{tikzpicture}
}
\subfigure[Spine of the recursion tree] {

\begin{tikzpicture}
[nodeDecorate/.style={shape=circle,inner sep=4pt,draw,fill=white},%
  lineDecorate/.style={-,dashed},%
  elipseDecorate/.style={color=gray!30},
  scale=0.35]

  \node (a) at (0,0) [nodeDecorate,label=above:$B_{u,t}$] {};
  \node (b) at (4,-1) [nodeDecorate] {};
  \node (c) at (8,-2) [nodeDecorate] {};
  \node (d) at (12,-3) [nodeDecorate] {};
  \node (e) at (16,-4) [nodeDecorate] {};

  \node (a2) at (-2,-5) [nodeDecorate,label=below:$B_{z_1,t}$] {};
  \node (b2) at (2,-6) [nodeDecorate,label=below:$B_{z_2,t}$] {};
  \node (c2) at (6,-7) [nodeDecorate,label=below:$B_{z_3,t}$] {};
  \node (d2) at (10,-8) [nodeDecorate,label=below:$B_{z_4,t}$] {};
  \node (e2) at (14,-9) [nodeDecorate,label=below:$B_{v,t}$] {};

\path {
	(a) edge node {} (b)
	(b) edge node {} (c)
	(c) edge node {} (d)
	(d) edge node {} (e)
	(a) edge node {} (a2)
	(b) edge node {} (b2)
	(c) edge node {} (c2)
	(d) edge node {} (d2)
	(e) edge node {} (e2)
};
\end{tikzpicture}
}
\caption{Example bead string and spine of the recursion tree}
\label{fig:topspine}
\end{figure}
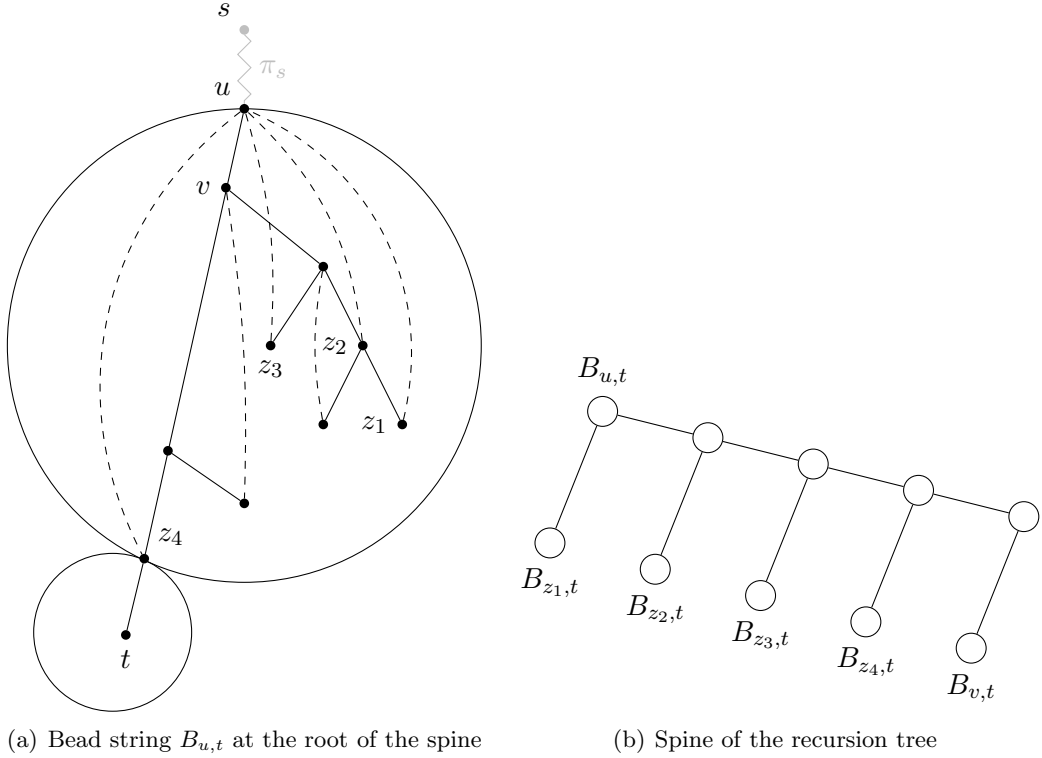

In this appendix, we present all details and illustrate with figures
the operations $\oracleright(C,e)$ and $\oracleleft(C,e)$ that are
performed along a spine of the recursion tree. In order to better
detail the procedures in Lemma~\ref{lem:removebackedge} and
Lemma~\ref{lem:promotebackedge}, we divide them in smaller parts.
Fig.~\ref{fig:topspine} shows (a) an example of a bead string
$B_{u,t}$ at the first node of the spine and (b) the nodes of the
spine. This spine contains four binary nodes corresponding to the
back edges in $lb(u)$ and an unary node corresponding to the tree edge
$(u,v)$. Note that edges are taken in reverse DFS order as defined in
operation $\chooseedge(C,u)$.  

As a consequence of Lemma~\ref{lemma:beadstring}, the impact of
operations $\oracleright(C,e)$ and $\oracleleft(C,e)$ in the
certificate is restricted to the biconnected component of $u$. Thus we
mainly focus on maintaining the compacted head $H_X = (V_X,E_X)$ of
the bead string $B_{u,t}$.

%\setcounter{theorem}{\ref{lem:removebackedge}}
%\addtocounter{theorem}{-1}

\subsection{Operation \boldmath{$\oracleright(C,e)$} in a spine of the recursion tree} 

Let us now prove and illustrate the following lemma:
\begin{lemma}
	\emph{(Lemma~\ref{lem:removebackedge} restated)} In a spine of
	the recursion tree, operations $\oracleright(C,e)$ can be
	implemented in $O(|V_X|)$ total time.
\end{lemma}

	In the right branches along a spine, we remove all back edges
	in $lb(u)$. This is done by starting from the last edge in
	$lb(u)$, i.e. proceeding in reverse DFS order. In the example
	from Fig.~\ref{fig:topspine}, we remove the back edges
	$(z_1,u)$, $(z_2,u)$, $(z_3,u)$ and $(z_4,u)$. To update the
	certificate corresponding to $B_{u,t}$, we have to (i) update
	the lowpoints in each vertex of $H_X$ (ii) prune vertices
	that are no longer in $B_{u,t}$ after removing a back edge.
	Note that, for a vertex $w$ in the tree, there is no need to
	update $\gamma(w)$.

	Let us consider the update of lowpoints in the DFS tree.  For
	a back edge $b_i = (z_i,u)$, we traverse the vertices in the
	path from $z_i$ towards the root $u$. By definition of
	lowpoint, these are the only lowpoints that can change.
	Suppose that we remove back edge $(z_4,u)$ in the example from
	Fig.~\ref{fig:topspine}, only the lowpoints of the vertices
	in the path from $z_4$ towards the root $u$ change.
	Furthermore, consider a vertex $w$ in the tree that is an
	ancestor of at least two endpoints $z_i, z_j$ of back edges
	$b_i$, $b_j$. The lowpoint of $w$ does not change when we
	remove $b_i$.  These observations lead us to the following
	lemma.

\begin{lemma}
	In a spine of the recursion tree, the update of lowpoints in
	the certificate by operation $\oracleright(C,e)$ can be done
	in $O(|V_X|)$ total time.  \label{lem:lowpoints}
\end{lemma}
\begin{proof}
	Take each back edge $b_i = (z_i,u)$ in the order defined by
	$\chooseedge(C,u)$.  Remove $b_i$ from $lb(u)$ and $ab(z_i)$.
	Starting from $z_i$, consider each vertex $w$ in the path from
	$z_i$ towards the root $u$.
	On vertex $w$, we update $\mathit{lowpoint}(w)$ using the standard
	procedure: take the endpoint $y$ of the first edge in $ab(w)$
	(the back edge that goes the nearest to the root of the tree)
	and choosing the minimum between $\gamma(y)$ and the lowpoint
	of each child of $w$.
	When the updated $\mathit{lowpoint}(w) = \gamma(u)$, we stop examining
	the path from $z_i$ to $u$ since it implies that the lowpoint
	of the vertex can not be further reduced (i.e. $w$ is both an
	ancestor to both $z_i$ and $z_{i+1}$).

	The total cost of updating the lowpoints when removing all
	back edges $b_i$ is $O(|V_X|)$: there are $O(|V_X|)$ tree
	edges and we do not traverse the same tree edge twice since
	the process described stops at the first common ancestor of
	endpoints of back edges $b_i$ and $b_{i+1}$. By contradiction:
	if a tree edge $(x,y)$ would be traversed twice when removing
	back edges $b_i$ and $b_{i+1}$, it would imply that both $x$
	and $y$ are ancestors of $z_i$ and $z_{i+1}$ (as edge $(x,y)$
	is both in the path $z_i$ to $u$ and the path $z_{i+1}$ to
	$u$) but we stop at the first ancestor of both $z_i$ and
	$z_{i+1}$.\qed
\end{proof}

Let us now consider the removal from the certificate of vertices that
are no longer in $B_{u,t}$ as consequence of operation
$\oracleright(C,e)$ in a spine of the recursion tree. By removing a
back edge $b_i = (z_i,u)$, it is possible that a vertex $w$ previously
in $H_X$ is no longer in the bead string $B_{u,t}$ (e.g. $w$ is no
longer biconnected to $u$ and thus there is no simple path $u \leadsto
w \leadsto t$). 

\begin{lemma}
	In a spine of the recursion tree, the branches of the DFS that
	are no longer in $B_{u,t}$ due to operation
	$\oracleright(C,e)$ can be removed from the certificate in
	$O(|V_X|)$ total time.
	\label{lem:cutbranches}
\end{lemma}
\begin{proof}
	To prune the branches of the DFS tree that are no longer in
	$H_X$, consider again each vertex $w$ in the path from $z_i$
	towards the root $u$ and the vertex $y$, parent of $w$. It is
	easy to check if $w$ is an articulation point by checking if
	the updated $\mathit{lowpoint}(w) > \gamma(y)$.  If that is the case
	and $w$ is not in the leftmost path of the DFS, it implies
	that $w \notin B_{u,t}$, and therefore we cut the subtree of
	$w$ and bookkeep it in $I$ to restore later. Like in the
	update the lowpoints, we stop examining the path $z_i$ towards
	$u$ in a vertex $w$ when $\mathit{lowpoint}(w) = \gamma(u)$ (the
	lowpoints and biconnected components in the path from $w$ to
	$u$ do not change).  When cutting the subtree of $w$, note
	that there are no back edges connecting it to $B_{u,t}$ ($w$ is
	an articulation point) and therefore there are no updates to
	the lists $lb$ and $ab$ of the vertices in $B_{u,t}$.  Like in
	the case of updating the lowpoints, we do not traverse the
	same tree edge twice (we use the same halting criterion). \qed
\end{proof}

With Lemma~\ref{lem:lowpoints} and Lemma~\ref{lem:cutbranches} we
finalize the proof of Lemma~\ref{lem:removebackedge}. 
Fig.~\ref{fig:oracleleftexample} shows the changes the bead string $B_{u,t}$
from Fig.~\ref{fig:topspine} goes through in the corresponding spine
of the recursion tree.

\begin{figure}[p]
\centering
%%%%%%%% ROOT
\subfigure[$B_{u,t}$ after removal of $(z_1,u)$]{
\begin{tikzpicture}
[nodeDecorate/.style={shape=circle,inner sep=1pt,draw,thick,fill=black},%
  lineDecorate/.style={-,dashed},%
  elipseDecorate/.style={color=gray!30},
  scale=0.35]
%bicon
\draw (10,22) circle (9);
\draw (5,11.1) circle (3);

\node (s) at (10,34) [nodeDecorate,color=lightgray,label=above left:$s$] {};
\node (u) at (10,31) [nodeDecorate,label=above left:{ $u$}] {};

\node (tp) at (6.2,13.9) [nodeDecorate,label=above right:{ $z_4$}] {};
\node (t) at (5.5,11) [nodeDecorate,label=below:$t$] {};

\path {
	(s) edge[snake,-,color=lightgray] node {\quad\quad$\pi_s$} (u)
	(u) edge node {} (tp)
	(tp) edge node {} (t)
};

\node (a) at (9.3,28) [nodeDecorate,label=left:$v$] {};
\node (b) at (7.1,18) [nodeDecorate,] {};
\node (c) at (13,25) [nodeDecorate,] {};
\node (d) at (14.5,22) [nodeDecorate,label=left:$z_2$] {};
\node (e) at (11,22) [nodeDecorate,label=below:$z_3$] {};
\node (g) at (13,19) [nodeDecorate,] {};
\node (h) at (10,16) [nodeDecorate,] {};

\path {
	(a) edge node {} (c)
	(c) edge node {} (d)
	(c) edge node {} (e)
	(d) edge node {} (g)
	(b) edge node {} (h)
};

\path {
	(u) edge[dashed,bend left=-40] node {} (tp)
	(g) edge[dashed,bend left=10] node {} (c)
	(e) edge[dashed,bend left=-10] node {} (u)
	(d) edge[dashed,bend left=-20] node {} (u)
	(a) edge[dashed,bend left=5] node {} (h)
};

\end{tikzpicture}
}
\subfigure[$B_{u,t}$ after removal of $(z_1,u),(z_2,u)$]{
\begin{tikzpicture}
[nodeDecorate/.style={shape=circle,inner sep=1pt,draw,thick,fill=black},%
  lineDecorate/.style={-,dashed},%
  elipseDecorate/.style={color=gray!30},
  scale=0.35]
%bicon
\draw (10,22) circle (9);
\draw (5,11.1) circle (3);

\node (s) at (10,34) [nodeDecorate,color=lightgray,label=above left:$s$] {};
\node (u) at (10,31) [nodeDecorate,label=above left:{ $u$}] {};

\node (tp) at (6.2,13.9) [nodeDecorate,label=above right:{ $z_4$}] {};
\node (t) at (5.5,11) [nodeDecorate,label=below:$t$] {};

\path {
	(s) edge[snake,-,color=lightgray] node {\quad\quad$\pi_s$} (u)
	(u) edge node {} (tp)
	(tp) edge node {} (t)
};

\node (a) at (9.3,28) [nodeDecorate,label=left:$v$] {};
\node (b) at (7.1,18) [nodeDecorate,] {};
\node (c) at (13,25) [nodeDecorate,] {};
\node (e) at (11,22) [nodeDecorate,label=below:$z_3$] {};
\node (h) at (10,16) [nodeDecorate,] {};

\path {
	(a) edge node {} (c)
	(c) edge node {} (e)
	(b) edge node {} (h)
};

\path {
	(u) edge[dashed,bend left=-40] node {} (tp)
	(e) edge[dashed,bend left=-10] node {} (u)
	(a) edge[dashed,bend left=5] node {} (h)
};

\end{tikzpicture}
}
\subfigure[$B_{u,t}$ after removal of $(z_1,u), (z_2,u), (z_3,u)$]{
\begin{tikzpicture}
[nodeDecorate/.style={shape=circle,inner sep=1pt,draw,thick,fill=black},%
  lineDecorate/.style={-,dashed},%
  elipseDecorate/.style={color=gray!30},
  scale=0.35]
%bicon
\draw (10,22) circle (9);
\draw (5,11.1) circle (3);

\node (s) at (10,34) [nodeDecorate,color=lightgray,label=above left:$s$] {};
\node (u) at (10,31) [nodeDecorate,label=above left:{ $u$}] {};

\node (tp) at (6.2,13.9) [nodeDecorate,label=above right:{ $z_4$}] {};
\node (t) at (5.5,11) [nodeDecorate,label=below:$t$] {};

\path {
	(s) edge[snake,-,color=lightgray] node {\quad\quad$\pi_s$} (u)
	(u) edge node {} (tp)
	(tp) edge node {} (t)
};

\node (a) at (9.3,28) [nodeDecorate,label=left:$v$] {};
\node (b) at (7.1,18) [nodeDecorate,] {};
\node (h) at (10,16) [nodeDecorate,] {};

\path {
	(b) edge node {} (h)
};

\path {
	(u) edge[dashed,bend left=-40] node {} (tp)
	(a) edge[dashed,bend left=5] node {} (h)
};

\end{tikzpicture}
}
\subfigure[$B_{u,t}$ at the last node of the spine]{

\begin{tikzpicture}
[nodeDecorate/.style={shape=circle,inner sep=1pt,draw,thick,fill=black},%
  lineDecorate/.style={-,dashed},%
  elipseDecorate/.style={color=gray!30},
  scale=0.35]
%bicon
\draw (10,22) circle (9) [color=gray!30];
\draw (5,11.1) circle (3);
\draw (9.6,29.5) circle (1.5);
\draw (8,23) circle (5);
\draw (6.7,15.9) circle (2.1);

\node (s) at (10,34) [nodeDecorate,color=lightgray,label=above left:$s$] {};
\node (u) at (10,31) [nodeDecorate,label=above left:{ $u$}] {};

\node (tp) at (6.2,13.9) [nodeDecorate,label=above right:{ $z_4$}] {};
\node (t) at (5.5,11) [nodeDecorate,label=below:$t$] {};

\path {
	(s) edge[snake,-,color=lightgray] node {\quad\quad$\pi_s$} (u)
	(u) edge node {} (tp)
	(tp) edge node {} (t)
};

\node (a) at (9.3,28) [nodeDecorate,label=below left:$v$] {};
\node (b) at (7.1,18) [nodeDecorate,] {};
\node (h) at (10,19) [nodeDecorate,] {};

\path {
	(b) edge node {} (h)
};

\path {
	(a) edge[dashed,bend left=5] node {} (h)
};

\end{tikzpicture}
}

\caption{Example application of $\oracleright(C,e)$ on a spine of the recursion tree}
\label{fig:oracleleftexample}
\end{figure}
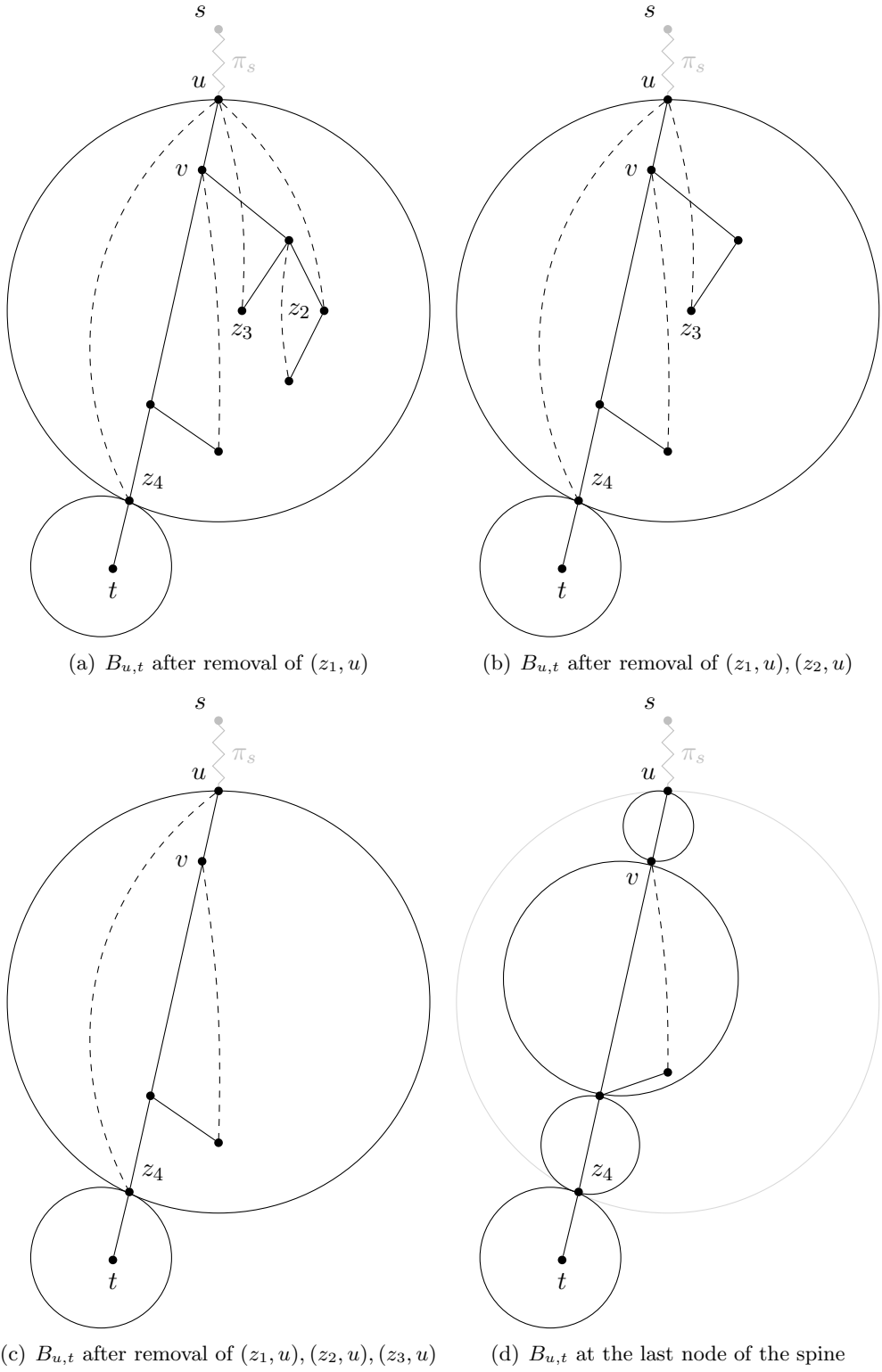

\subsection{Operation \boldmath{$\oracleleft(C,e)$} in a spine of the recursion tree} 

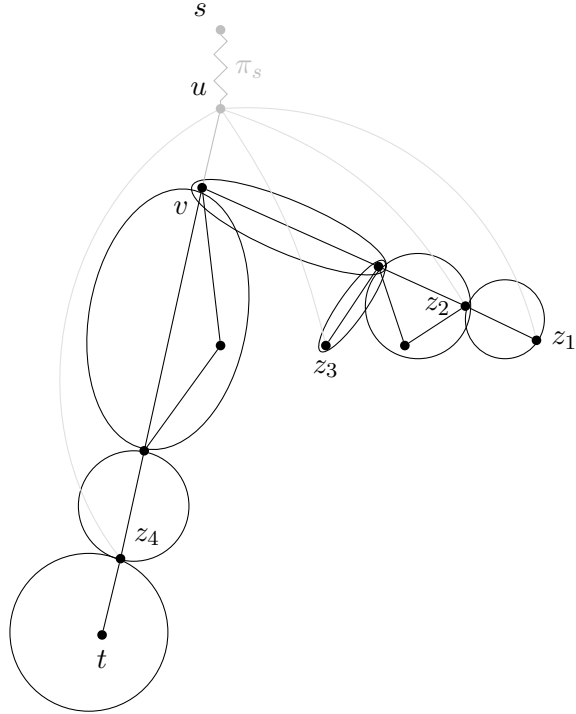
\begin{figure}[t!]
\centering
%%%%%%%% ROOT
\begin{tikzpicture}
[nodeDecorate/.style={shape=circle,inner sep=1pt,draw,thick,fill=black},%
  lineDecorate/.style={-,dashed},%
  elipseDecorate/.style={color=gray!30},
  scale=0.35]
%bicon
\draw (5,11.1) circle (3);
\draw[rotate around={-10:(8,23)}] (8,23) ellipse (3 and 5);
\draw[rotate around={-23:(12.6,26.5)}] (12.6,26.5) ellipse (4 and 1);
\draw[rotate around={55:(15,23.5)}] (15,23.5) ellipse (2.1 and 0.5);
\draw (6.7,15.9) circle (2.1);
\draw (17.5,23.5) circle (2.0);
\draw (20.8,23) circle (1.5);

\node (s) at (10,34) [nodeDecorate,color=lightgray,label=above left:$s$] {};
\node (u) at (10,31) [nodeDecorate,color=lightgray,label=above left:{ $u$}] {};

\node (tp) at (6.2,13.9) [nodeDecorate,label=above right:{ $z_4$}] {};
\node (t) at (5.5,11) [nodeDecorate,label=below:$t$] {};

\path {
	(s) edge[snake,-,color=lightgray] node {\quad\quad$\pi_s$} (u)
	(u) edge[color=lightgray] node {} (a)
	(a) edge node {} (tp)
	(tp) edge node {} (t)
};

\node (a) at (9.3,28) [nodeDecorate,label=below left:$v$] {};
\node (b) at (7.1,18) [nodeDecorate,] {};
\node (c) at (16,25) [nodeDecorate,] {};
\node (d) at (19.3,23.5) [nodeDecorate,label=left:$z_2$] {};
\node (e) at (14,22) [nodeDecorate,label=below:$z_3$] {};
\node (f) at (22,22.2) [nodeDecorate,label=right:$z_1$] {};
\node (g) at (17,22) [nodeDecorate,] {};
\node (h) at (10,22) [nodeDecorate,] {};

\path {
	(a) edge node {} (c)
	(c) edge node {} (d)
	(d) edge node {} (f)
	(c) edge node {} (e)
	(d) edge node {} (g)
	(b) edge node {} (h)
};

\path {
	(g) edge node {} (c)
	(a) edge node {} (h)


	(u) edge[bend left=-50,color=lightgray!50] node {} (tp)
	(f) edge[bend left=-40,color=lightgray!50] node {} (u)
	(e) edge[bend left=-10,color=lightgray!50] node {} (u)
	(d) edge[bend left=-20,color=lightgray!50] node {} (u)
};

\end{tikzpicture}
\caption{Block tree after removing vertex $u$}
\label{fig:blocktree_without_u}
\end{figure}

To compute $\oracleleft(C,e)$ in the binary nodes of a spine, we use
the fact that in every left branching from that spine, the graph is
the same (in a spine we only remove edges incident to $u$ and on a
left branch from the spine we remove the node $u$) and therefore its
block tree is also the same. In Fig.~\ref{fig:blocktree_without_u},
we show the resulting block tree of the graph from
Fig.~\ref{fig:topspine} after having removed vertex $u$. However,
the certificates on these left branches are not the same, as they are
rooted at different vertices. In the example we must compute the
certificates $C_1 \ldots C_4$ corresponding to bead strings $B_{z_1,t}
\ldots B_{z_4,t}$. We do not take into account the cost of the left
branch on the last node of spine (corresponding to $B_{v,t}$) as the
node is unary and we have shown in Lemma~\ref{lem:unary_left} how to
maintain the certificate in $O(1)$ time.

By using the reverse DFS order of the back edges, we are able to
traverse each edge in $H_X$ only an amortized constant number of times
in the spine.
\begin{lemma}
	\emph{(Lemma~\ref{lem:promotebackedge} restated)} The calls to
	operation $\oracleleft(C,e)$ in a spine of the recursion tree
	can be charged with a time cost of $O(|E_X|)$ to that spine. 
\end{lemma}

To achieve this time cost, for each back edge $b_i = (z_i,u)$, we
compute the certificate corresponding to $B_{z_i,t}$ based on the
certificate of $B_{z_{i-1},t}$. Consider the compacted head $H_X =
(V_X , E_X )$ of the bead string $B_{u,t}$. We use $O(|E_X|)$ time to
compute the first certificate $C_1$ that corresponds to bead string
$B_{z_1,t}$. Fig.~\ref{fig:spine-left-update} shows bead string
$B_{z_1,t}$ from the example of Fig.~\ref{fig:topspine}.

\begin{lemma}
	The certificate $C_1$, corresponding to bead string
	$B_{z_1,t}$, can be computed in $O(|E_X|)$ time.
\end{lemma}
\begin{proof}
	Let $t'$ be the last vertex in the path $u \leadsto t$ s.t.
	$t' \in V_X$. Since $t'$ is an articulation point, the subtree
	of the DFS tree rooted in $t'$ is maintained in the case of
	removal vertex $u$. Therefore the only modifications of the
	DFS tree occur in head $H_X$ of $B_{u,t}$.

	To compute $C_1$, we remove $u$ and rebuild the certificate
	starting form $z_1$ using the algorithm from
	Lemma~\ref{lem:certificate_scratch} restricted to $H_X$ and
	using $t'$ as target and $\gamma(t')$ as a baseline to
	$\gamma$ (instead of the depth). In particular we do
	the following.
	In order to set $t'$ to be in the leftmost path, we perform a
	DFS traversal of graph $H_X$ starting from $z_1$ and stop when
	we reach vertex $t'$. We then compute the DFS tree, traversing
	the path $z_1 \leadsto t'$ first.
	
	{\it Update of $\gamma$.} For each tree edge $(v,w)$ in the
	$t' \leadsto z_1$ path, we set $\gamma(v)=\gamma(w)-1$, using
	$\gamma(t')$ as a baseline.  During the rest of the traversal,
	when visiting vertex $v$, let $w$ be the parent of $v$ in the
	DFS tree. We set $\gamma(v)=\gamma(w)+1$. This maintains the
	property that $\gamma(v)>\gamma(w)$ for any $w$ ancestor of
	$v$.

	{\it Update of lowpoints and pruning of the tree.}  Bottom-up
        in the DFS-tree of $H_X$, we compute the lowpoints using the
        lowpoints of the children.  Let $z$ be the parent of $v$. If
        $\mathit{lowpoint}(v) > \mathit{lowpoint}(z)$ and $v$ is not
        in the leftmost path in the DFS, we cut the subtree of $v$ as
        it does not belong to $B_{z_1,t}$.

	{\it Computing $lb$ and $ab$.} During the traversal, when
	finding a back edge $e=(v,w)$, if $w$ is a descendant of $v$
	we append $e$ to both $lb(v)$ and $ab(w)$; else we append $e$
	to both $ab(v)$ and $lb(w)$.  This maintains the DFS order in
	the back edge lists.

	This procedure takes $O(|E_X|)$ time.
\end{proof}

To compute each certificate $C_i$, corresponding to bead string
$B_{z_i,t}$, we are able to avoid visiting most of the edges that
belong $B_{z_{i-1},t}$. Since we take $z_i$ in reverse DFS order, on
the spine of the recursion we visit $O(|E_X|)$ edges plus a term that
can be amortized.

\begin{lemma}
	For each back edge $b_i = (z_i,u)$ with $i>1$, let ${E_X}_i'$
	be the edges in the first bead in common between $B_{z_i,t}$
	and $B_{z_{i-1},t}$. The total cost of computing all
	certificates $B_{z_i,t}$ in a spine of the recursion tree is:
	$O(E_X + \sum_{i>1}{{E_X}_i'})$.
	\label{lem:cost-spine-not-amortized}
\end{lemma}
\begin{proof}
	Let us compute the certificate $C_i$: this is the certificate
	of the left branch of the $i$th node of the spine where we
	augment the path with the back edge $b_i = (z_i,u)$ of
	$lb(u)$.

	For the general case of $C_i$ with $i>1$ we also rebuild
	(part) of the certificate starting from $z_i$ using the
	procedure from Lemma~\ref{lem:certificate_scratch} but we use
	information gathered in $C_{i-1}$ to avoid exploring useless
	branches of the DFS tree. The key point is that, when we reach
	the first bead in common to both $B_{z_i,t}$ and
	$B_{z_{i-1},t}$, we only explore edges internal to this bead.
	If an edge $e$ that leaves the bead leads to $t$, we can reuse
	a subtree of $C_{i-1}$. If $e$ does not lead to $t$, then it
	has already been explored (and cut) in $C_{i-1}$ and there is
	no need to explore it again since it is going to be discarded.

	In detail, we start computing a DFS from $z_i$ in $B_{u,t}$
	until we reach a node $t' \in B_{z_{i-1},t}$. Note that the
	bead of $t'$ has one entry point and one exit point in
	$C_{i-1}$. After reaching $t'$ we proceed with the traversal
	using only edges already in $C_{i-1}$. When arriving at a
	vertex $w$ that is not in the same bead of $t'$, we stop the
	traversal. If $w$ is in a bead towards $t$, we reuse the
	subtree of $w$ and use $\gamma(w)$ as a baseline of the
	numbering $\gamma$. Otherwise $w$ is in a bead towards
	$z_{i-1}$ and we cut this branch of the certificate. When all
	edges in the bead of $t'$ are traversed, we proceed with visit
	in the standard way.

	Given the order we take $b_i$, each bead is not added more
	than once to a certificate $C_i$, therefore the total cost
	over the spine is $O(|E_X|)$.
	Nevertheless, the internal edges ${E_X}_i'$ of the first bead
	in common between $B_{z_i,t}$ and $B_{z_{i-1},t}$ are explored
	for each back edge $b_i$.\qed
\end{proof}

Although the edges in ${E_X}_i'$ are in a common bead between
$B_{z_i,t}$ and $B_{z_{i-1},t}$, these edges must be visited. Since
the entry point in the common bead can be different for $z_i$ and
$z_{i-1}$, the DFS tree of that bead can also be different. For an
example, consider the case where $z_i, \ldots, z_j$ are all in the
same bead after the removal of $u$. The bead strings $B_{z_i,t} \ldots
B_{z_j,t}$ are the same, but the roots $z_i, \ldots, z_j$ of the
certificate are different, so we have to compute the corresponding DFS
of the first bead $|j-i|$ times. Note that this is not the case for
the other beads in common: the entry point is always the same.

\begin{lemma}
	The cost $O(E_X + \sum_{i>1}{{E_X}_i'})$ on a spine of the
	recursion tree can be amortized to $O(E_X)$ in that spine.
	\label{lem:left-amortize}
\end{lemma}
\begin{proof}
	We can charge the cost $O(|{E_X}_i'|)$ of exploring the edges
	in the first bead in common between $B_{z_i,t}$ and
	$B_{z_{i-1},t}$ to another node in the recursion tree. Since
	this common bead is the head of at least one certificate in
	the recursion subtree of the left child of the $i$th node of
	the spine.  Specifically, we charge the first and only node in
	the \emph{leftmost} path of the $i$th child of the spine that
	has exactly the edges ${E_X}_i'$ as head of its bead string:
	(i) if $|{E_X}_i'| \le 1$ it corresponds to a unary node or a
	leaf in the recursion tree and therefore we can charge it with
	$O(1)$ cost; (ii) otherwise it corresponds to a first node of
	a spine and therefore we can also charge it with
	$O(|{E_X}_i'|)$. We use this charging scheme when $i \neq 1$
	and the cost is always charged in the leftmost recursion path
	of $i$th node of the spine, consequently we never charge a
	node in the recursion tree more than once. \qed
\end{proof}

Lemmas~\ref{lem:cost-spine-not-amortized} and~\ref{lem:left-amortize}
finalize the proof of Lemma~\ref{lem:promotebackedge}. 
Fig.~\ref{fig:spine-left-update} shows the certificates of bead strings
$B_{z_i,t}$ on the left branches of the spine from
Fig.~\ref{fig:topspine}.

\begin{figure}[t]
\centering
%%%%%%%% ROOT
\subfigure[Certificate of bead string $B_{z_1,t}$] {

\begin{tikzpicture}
[nodeDecorate/.style={shape=circle,inner sep=1pt,draw,thick,fill=black},%
  lineDecorate/.style={-,dashed},%
  elipseDecorate/.style={color=gray!30},
  scale=0.22]
%bicon
\draw (10,28) circle (2);
\draw[color=white] (10,14) circle (11);
\draw (10,23.5) circle (2.5);
\draw (10,19.5) circle (1.5);
\draw (10,14) circle (4);
\draw (10,8) circle (2);
\draw (10,3) circle (3);

\node (s) at (10,34) [nodeDecorate,color=lightgray,label=above left:$s$] {};

\node (z1) at (10,30) [nodeDecorate,label=below right:$z_1$] {};
\node (z2) at (10,26) [nodeDecorate,label=below left:$z_2$] {};
\node (a) at (10,21) [nodeDecorate] {};
\node (b) at (12,23) [nodeDecorate] {};
\node (v) at (10,18) [nodeDecorate,label=below left:$v$] {};
\node (c) at (10,14) [nodeDecorate] {};
\node (d) at (10,10) [nodeDecorate] {};
\node (z4) at (10,6) [nodeDecorate,label=above right:$z_4$] {};
\node (t) at (10,2) [nodeDecorate,label=right:$t$] {};

\path {
	(s) edge[snake,-,color=lightgray] node {\quad\quad$\pi_s$} (z1)
	(z1) edge node {} (z2)
	(z2) edge node {} (a)
	(a) edge node {} (b)
	(a) edge node {} (v)
	(v) edge node {} (c)
	(c) edge node {} (d)
	(d) edge node {} (z4)
	(z4) edge node {} (t)
};

\path {
	(d) edge[dashed,bend left=-50] node {} (v)
	(b) edge[dashed,bend left=-30] node {} (z2)
};
\end{tikzpicture}
}
\subfigure[Certificate of bead string $B_{z_2,t}$] {

\begin{tikzpicture}
[nodeDecorate/.style={shape=circle,inner sep=1pt,draw,thick,fill=black},%
  lineDecorate/.style={-,dashed},%
  elipseDecorate/.style={color=gray!30},
  scale=0.22]
%bicon
\draw[color=white] (10,14) circle (11);
\draw (10,23.5) circle (2.5);
\draw (10,19.5) circle (1.5);
\draw (10,14) circle (4);
\draw (10,8) circle (2);
\draw (10,3) circle (3);

\node (s) at (10,30) [nodeDecorate,color=lightgray,label=above left:$s$] {};

\node (z2) at (10,26) [nodeDecorate,label=below left:$z_2$] {};
\node (a) at (10,21) [nodeDecorate] {};
\node (b) at (12,23) [nodeDecorate] {};
\node (v) at (10,18) [nodeDecorate,label=below left:$v$] {};
\node (c) at (10,14) [nodeDecorate] {};
\node (d) at (10,10) [nodeDecorate] {};
\node (z4) at (10,6) [nodeDecorate,label=above right:$z_4$] {};
\node (t) at (10,2) [nodeDecorate,label=right:$t$] {};

\path {
	(s) edge[snake,-,color=lightgray] node {\quad\quad$\pi_s$} (z2)
	(z2) edge node {} (a)
	(a) edge node {} (b)
	(a) edge node {} (v)
	(v) edge node {} (c)
	(c) edge node {} (d)
	(d) edge node {} (z4)
	(z4) edge node {} (t)
};

\path {
	(d) edge[dashed,bend left=-50] node {} (v)
	(b) edge[dashed,bend left=-30] node {} (z2)
};
\end{tikzpicture}
}
\subfigure[Certificate of bead string $B_{z_3,t}$] {

\begin{tikzpicture}
[nodeDecorate/.style={shape=circle,inner sep=1pt,draw,thick,fill=black},%
  lineDecorate/.style={-,dashed},%
  elipseDecorate/.style={color=gray!30},
  scale=0.22]
%bicon
\draw[color=white] (10,14) circle (11);
\draw (10,23.5) circle (2.5);
\draw (10,19.5) circle (1.5);
\draw (10,14) circle (4);
\draw (10,8) circle (2);
\draw (10,3) circle (3);

\node (s) at (10,30) [nodeDecorate,color=lightgray,label=above left:$s$] {};

\node (z3) at (10,26) [nodeDecorate,label=below left:$z_3$] {};
\node (a) at (10,21) [nodeDecorate] {};
\node (v) at (10,18) [nodeDecorate,label=below left:$v$] {};
\node (c) at (10,14) [nodeDecorate] {};
\node (d) at (10,10) [nodeDecorate] {};
\node (z4) at (10,6) [nodeDecorate,label=above right:$z_4$] {};
\node (t) at (10,2) [nodeDecorate,label=right:$t$] {};

\path {
	(s) edge[snake,-,color=lightgray] node {\quad\quad$\pi_s$} (z3)
	(z3) edge node {} (a)
	(a) edge node {} (v)
	(v) edge node {} (c)
	(c) edge node {} (d)
	(d) edge node {} (z4)
	(z4) edge node {} (t)
};

\path {
	(d) edge[dashed,bend left=-50] node {} (v)
};
\end{tikzpicture}
}

\subfigure[Certificate of bead string $B_{z_4,t}$] {

\begin{tikzpicture}
[nodeDecorate/.style={shape=circle,inner sep=1pt,draw,thick,fill=black},%
  lineDecorate/.style={-,dashed},%
  elipseDecorate/.style={color=gray!30},
  scale=0.22]
%bicon
\draw[color=white] (10,8) circle (11);
\draw (10,3) circle (3);

\node (s) at (10,14) [nodeDecorate,color=lightgray,label=above left:$s$] {};

\node (z4) at (10,6) [nodeDecorate,label=above right:$z_4$] {};
\node (t) at (10,2) [nodeDecorate,label=right:$t$] {};

\path {
	(s) edge[snake,-,color=lightgray] node {\quad\quad$\pi_s$} (z4)
	(z4) edge node {} (t)
};

\end{tikzpicture}
}
\caption{Certificates of the left branches of a spine}
\label{fig:spine-left-update}
\end{figure}
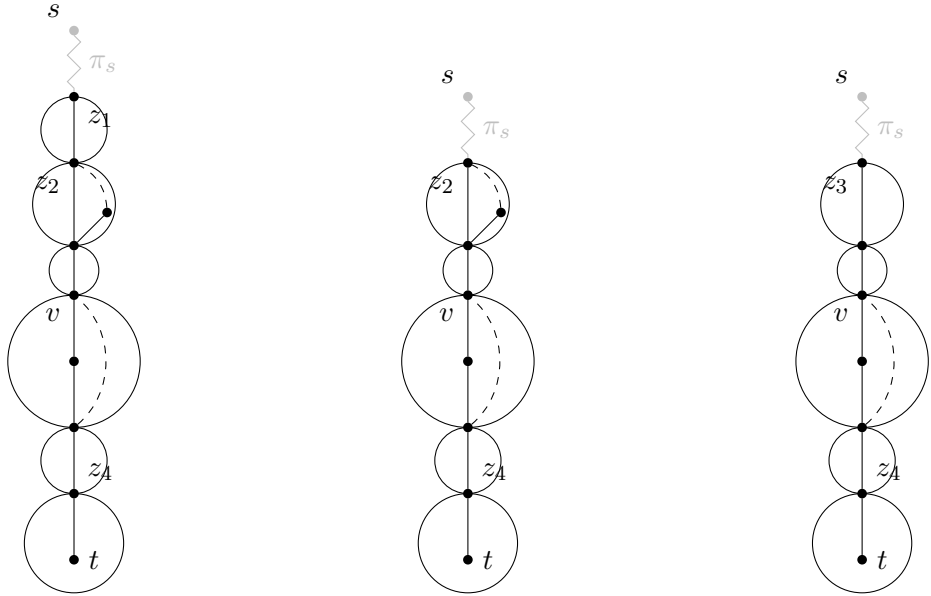
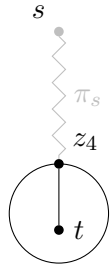

\end{document}